\documentclass[a4paper,11pt]{article}
\usepackage{bm,mathrsfs,amsmath,amsthm,amssymb,amscd,longtable,array,graphicx,color}
\newcommand{\nc}{\newcommand}
\nc{\rnc}{\renewcommand}
\nc{\nn}{\nonumber}
\nc{\der}{{\partial}}
\rnc{\Im}{{\rm{Im}\,}}
\rnc{\Re}{{\rm{Re}\,}}
\nc{\db}{\displaybreak[0]\\}
\nc{\bra}{\langle}
\nc{\ket}{\rangle}
\nc{\bs}{\boldsymbol}

\newtheorem{theorem}{Theorem}[section]
\newtheorem{lemma}[theorem]{Lemma}

\newtheorem{proposition}[theorem]{Proposition}
\newtheorem{corollary}[theorem]{Corollary}

\theoremstyle{definition}

\numberwithin{equation}{section}

\numberwithin{equation}{section}

\begin{document}%
%
\title{Tetrahedron equation and Schur functions}

\author{
Shinsuke Iwao, \thanks{Email address: iwao-s@keio.jp
{\it Faculty of Business and Commerce, Keio University,}
{\it Hiyosi 4–1–1, Kohoku-ku, Yokohama-si, Kanagawa 223-8521, Japan},
}
\ \ \
Kohei Motegi, \thanks{E-mail: kmoteg0@kaiyodai.ac.jp 
{\it Faculty of Marine Technology, Tokyo University of Marine Science and Technology,}
 {\it Etchujima 2-1-6, Koto-Ku, Tokyo, 135-8533, Japan},
}
\
and \
Ryo Ohkawa \thanks{E-mail: ohkawa.ryo@gmail.com, ohkawa.ryo@omu.ac.jp
{\it Osaka Central Advanced Mathematical Institute, Osaka Metropolitan
University, 558-8585, Osaka, Japan},
{\it Research Institute for Mathematical Sciences, Kyoto University,
606-8502, Kyoto, Japan.
}
}
}

\date{\today}

\maketitle

\begin{abstract}
The tetrahedron equation introduced by Zamolodchikov is a three-dimensional generalization of the Yang-Baxter equation.
Several types of solutions to the tetrahedron equation that have connections to quantum groups can be viewed as $q$-oscillator valued vertex models with matrix elements of the $L$-operators given by generators of the $q$-oscillator algebra acting on the Fock space.
Using one of the $q=0$-oscillator valued vertex models introduced by Bazhanov-Sergeev, we introduce a family of partition functions that admits an explicit algebraic presentation using Schur functions.
Our construction is based on the three-dimensional realization of the Zamolodchikov-Faddeev algebra provided by Kuniba-Okado-Maruyama. Furthermore, we investigate an inhomogeneous generalization of the three-dimensional lattice model. We show that the inhomogeneous analog of (a certain subclass of) partition functions can be expressed as loop elementary symmetric functions.
\end{abstract}

\section{Introduction}

In~\cite{Zamone,Zamtwo}, Zamolodchikov introduced the tetrahedron equation as a three-dimensional generalization of the Yang-Baxter equation \cite{BaxterYBE,Yang}.
Various seminal works such as \cite{Baxterone,Baxtertwo,BaSt,Kashaev,KKS,Kor} have been done
to construct solutions to the tetrahedron equation.
See \cite{BMS,BaSe,Kuniba,KMO1,KMO2,KOS,MBS}, as well as more recently \cite{KMY}, for examples of studies that investigate solutions to the tetrahedron equation through the viewpoint of quantum groups \cite{Drinfeld,FRT,FST,Jimbo}.

Solutions to the tetrahedron equation which we call as three-dimensional $R$-matrices are regarded as ``local Boltzmann weights" in three-dimensional statistical physics.
In this paper, we study a class of partition functions 
using one of the solutions found by Bazhanov-Sergeev~\cite{BaSe}, where
they used the three-dimensional $R$-matrix and gave an unconventional realization of the $U_q(\widehat{sl_n})$ $R$-matrix by tracing out one dimension. 
Our main idea is to use a degeneration of the three-dimensional $R$-matrix,
which was studied by Kuniba-Maruyama-Okado \cite{KMO1,KMO2} from the perspective of its relation to a multispecies generalization of a totally asymmetric symmetric exclusion process (TASEP) \cite{DEHP,Spitzer}.
They showed that the matrix product form of the steady state of the multispecies periodic TASEP introduced by Ferrari-Martin \cite{FM} can be transformed into a certain type of partition functions constructed from Bazhanov-Sergeev's three-dimensional $R$-matrix.
This result was given by bypassing the crystal basis theory \cite{Kashiwara,NY}.
Furthermore,
they showed that layers of the three-dimensional partition functions
satisfy the Zamolodchikov-Faddeev algebra relations, i.e.
they gave a realization of generators of the (degenerated) Zamolodchikov-Faddeev algebra.
The partition functions corresponding to the steady state were constructed by taking the traces of the product of the Zamolodchikov-Faddeev operators acting on the tensor product of Fock spaces and specializing spectral parameters to a specific value.

In this paper, we study a closely related but different class of three-dimensional partition functions.
Instead of taking the traces, our partition functions are given by taking the vacuum expectation values of the Zamolodchikov-Faddeev operators. 
In addition, we do not specialize spectral parameters to a specific value, as they will play the role of symmetric variables. 
We show that this class of partition functions is represented as Schur polynomials up to an overall factor.
This result is derived from multiple commutation relations of the Zamolodchikov-Faddeev algebra.

Furthermore, we also investigate an inhomogeneous generalization.
Since this generalization does not respect integrability, techniques based on integrable systems such as the
Zamolodchikov-Faddeev algebra cannot be applicable here. 
However, we show that a certain subclass of inhomogeneous partition functions can be expressed as loop elementary symmetric functions, which was originally introduced by Yamada \cite{Yamada} in the context of the discrete Toda lattice.
See Lam-Pylyavskyy \cite{Lam,LP} for more details.
We also show that several other partition functions are given by loop elementary symmetric functions.

This paper is organized as follows. 
Section \ref{sec:three-simensional_Rmatrix} first reviews some basic facts about the three-dimensional $R$-matrix and its equivalent description as an operator-valued two-dimensional $R$-matrix.
We next recall a realization of the Zamolodchikov-Faddeev operators, and introduce partition functions.
We derive the correspondence with the Schur functions in Section \ref{sec:Schur} and also discuss some corollaries/applications.
In section \ref{sec:inhomo}, we introduce an inhomogeneous version
for a subclass and several other related partition functions,
and prove the correspondence with the loop elementary symmetric functions.

\section{Three-dimensional $R$-matrix and partition functions}\label{sec:three-simensional_Rmatrix}
We first recall the three-dimensional $R$-matrix \cite{BaSe,KMO1,KMO2}
and operators that are used in this paper,
and introduce the partition functions which we investigate.

Let $V={\mathbb C} v_0 \oplus {\mathbb C} v_1$ be the complex two-dimensional space spanned by the standard basis $\{v_0,v_1 \}$, and let $\mathcal{F}=\displaystyle \oplus_{m=0}^\infty \mathbb{C} |m \rangle$
be the bosonic Fock space spanned by the basis $\{|m \rangle \;;\; (m=0,1,\dots) \}$.
We introduce the linear operators $\mathbf{t}, \mathbf{b}^+, \mathbf{b}^-$ acting on $\mathcal{F}$ as
\begin{align}
\mathbf{t}|m \rangle = \delta_{m,0}|m \rangle, \ \ \
\mathbf{b}^+|m \rangle =|m+1 \rangle, \ \ \
\mathbf{b}^-|m \rangle =|m-1 \rangle.
\end{align}
Here, $|-1 \rangle:=0$. 
These operators satisfy the relations
\begin{align}
\mathbf{t} \mathbf{b}^+=\mathbf{b}^- \mathbf{t}=0, \ \ \
\mathbf{b}^+ \mathbf{b}^-=1-\mathbf{t}, \ \ \ \mathbf{b}^- \mathbf{b}^+=1.
\end{align}
We refer to $\mathbf{t}, \mathbf{b}^+, \mathbf{b}^-$ as the zero-number projection operator, the creation operator, and the annihilation operator, respectively.
The bosonic Fock space and the algebra associated can be regarded as the $q=0$ limit of the $q$-bosonic Fock space
and the $q$-oscillator algebra.

Let $V^\ast=\mathbb{C} v^\ast_0\oplus \mathbb{C} v^\ast_1$ denote the dual space of $V^\ast$, which is spanned by $\{v_0^\ast,v_1^\ast\}$ such that $v_j^\ast(v_i)=\delta_{ij}$.
The dual Fock space $\mathcal{F}^\ast=\oplus_{m=0}^\infty \mathbb{C}\langle m|$ is the vector space spanned by $\{\langle m|\;|\;m=0,1,\dots\}$ such that $\langle m'\;|\;m\rangle=\delta_{m',m}$, where $(\langle m'|,|m\rangle)\mapsto \langle m'\;|\;m\rangle$ is the standard pairing.

The three-dimensional $R$-matrix $R$ acting on $ V \otimes V \otimes \mathcal{F}$
is defined by acting on the basis $\{ v_i \otimes v_j \otimes |k \rangle \;;\;(i,j=0,1, k=0,1,2,\dots) \}$ as
\begin{align}\label{eq:three_dimensional_R}
\displaystyle
R( v_i \otimes v_j \otimes |k \rangle)
=\sum_{a=0}^1 \sum_{b=0}^1 \sum_{\ell=0}^\infty  v_a \otimes v_b \otimes |\ell \rangle [R]_{ijk}^{ab \ell},
\end{align}
where
\begin{align}
[R]_{00k}^{00 \ell}&=\delta_{k \ell}, \ \ \ 
[R]_{11k}^{11 \ell}=\delta_{k \ell}, \ \ \
[R]_{10k}^{01 \ell}=\delta_{k+1, \ell}, \\
[R]_{01k}^{10 \ell}&=\delta_{k, \ell+1}, \ \ \
[R]_{01k}^{01 \ell}=\delta_{k 0} \delta_{\ell 0},
\end{align}
and $[R]_{ijk}^{ab \ell}=0$ otherwise.
This is equivalent to the operator-valued two-dimensional $L$-operator $\mathcal{R}$ acting on $ V \otimes V$
as
\begin{align}\label{eq:operator_valued_R}
\mathcal{R}( v_i \otimes v_j)
=\sum_{a=0}^1 \sum_{b=0}^1  v_a \otimes v_b [\mathcal{R}]_{ij}^{ab},
\end{align}
where
\begin{align}
[\mathcal{R}]_{00}^{00}&=1, \ \ \
[\mathcal{R}]_{11}^{11}=1, \ \ \
[\mathcal{R}]_{10}^{01}=\mathbf{b}^+, \\
[\mathcal{R}]_{01}^{10}&=\mathbf{b}^-, \ \ \
[\mathcal{R}]_{01}^{01}=\mathbf{t},
\end{align}
and $[\mathcal{R}]_{ij}^{ab}=0$ otherwise.

\begin{figure}[htbp]
\centering
\includegraphics[width=12truecm]{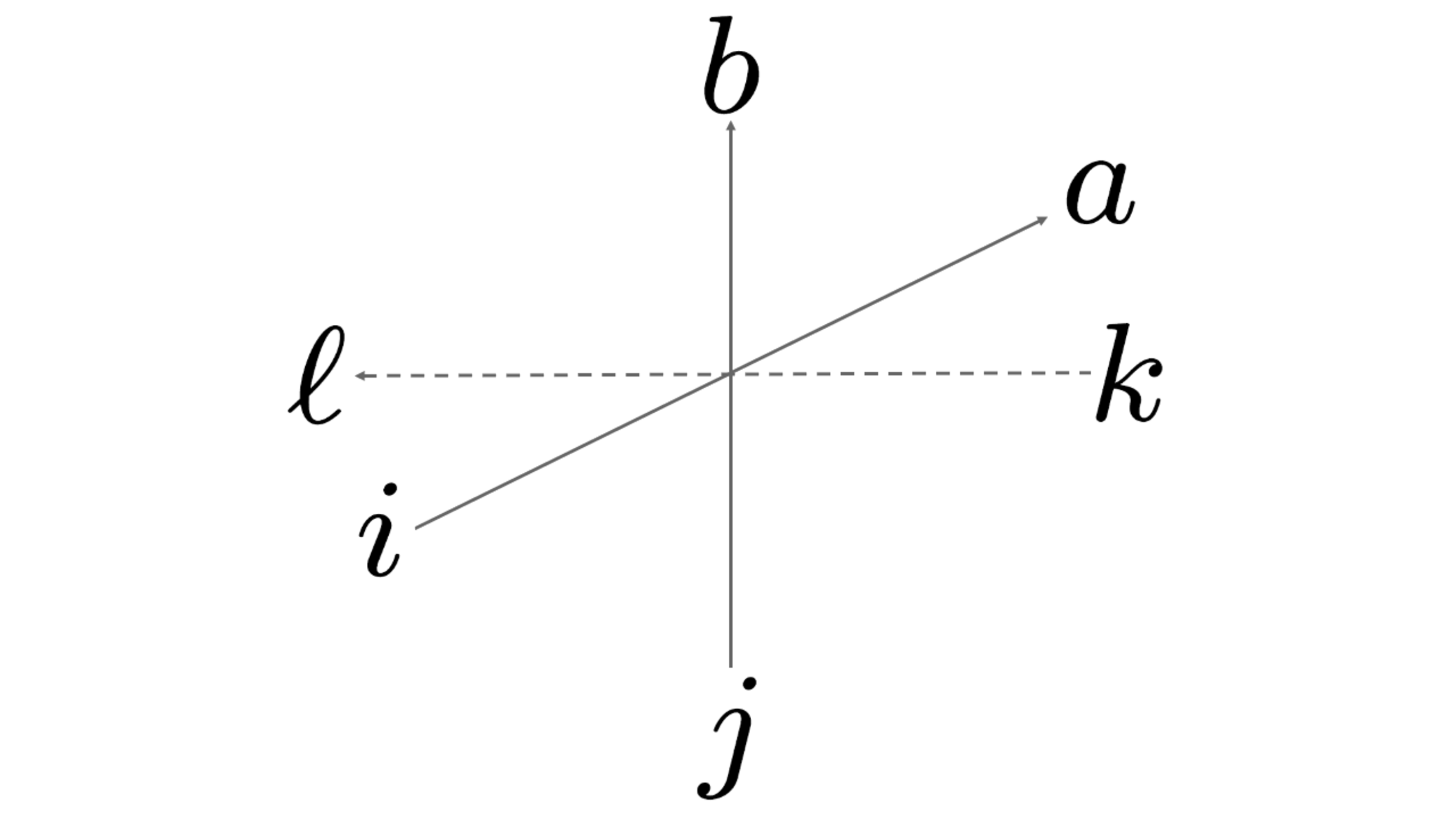}
\caption{Three-dimensional $R$-matrix. The two solid lines represent two two-dimensional spaces,
and the dashed line represents the bosonic Fock space.
To this configuration,
the "Boltzmann weight" $[R]_{ijk}^{ab \ell}$ is assigned.
One can also regard this is as an "operator-valued Boltzmann weight"
$[\mathcal{R}]_{ij}^{ab}$ if not specifying the basis in the bosonic Fock space.
}
\label{3DRfigure}
\end{figure}

Figure \ref{3DRfigure} presents a graphical description of the three-dimensional $R$-matrix $R$ in \eqref{eq:three_dimensional_R} (or equivalently, the operator-valued two-dimensional $R$-matrix $\mathcal{R}$ in  \eqref{eq:operator_valued_R}), which we use in this paper.
The matrix elements of $\mathcal{R}$ are shown in Figure \ref{3DRtwofigure}.

Let $\mathbb{Z}^3=\mathbb{Z}\bm{e}_1\oplus \mathbb{Z}\bm{e}_2\oplus\mathbb{Z}\bm{e}_3$ be the three-dimensional lattice spanned by the standard basis $\{\bm{e}_1,\bm{e}_2,\bm{e}_3\}$.
An element $\alpha\bm{e}_1+\beta \bm{e}_2+\gamma\bm{e}_3\in \mathbb{Z}^3$ is often identified with the vector $(\alpha,\beta,\gamma)$.
We will refer to the planes 
$H^1_k=\{(\alpha,\beta,\gamma)\;|\; \alpha=k\}$,
$H^2_{\ell}=\{(\alpha,\beta,\gamma)\;|\; \beta=\ell\}$, and
$H^3_m=\{(\alpha,\beta,\gamma)\;|\; \gamma=m\}$ as
the $k$-th row, the $\ell$-th column, and the $m$-th slice, respectively.

Let $D_n:=\{(k,\ell)\in \mathbb{Z}^2\;|\;k\geq 1,\ell\geq 1,k+\ell\leq n\}$ be the set consisting of $\frac{n(n-1)}{2}$ points.
We assign a copy of the Fock space $\mathcal{F}$ to each line $L_{k\ell}:=H^1_k\cap H^2_{\ell}$ for $(k,\ell)\in D_n$.
Hereinafter, $\mathcal{F}_{k\ell}$ denotes the Fock space assigned to $L_{k\ell}$.

\begin{figure}[htbp]
\centering
\includegraphics[width=12truecm]{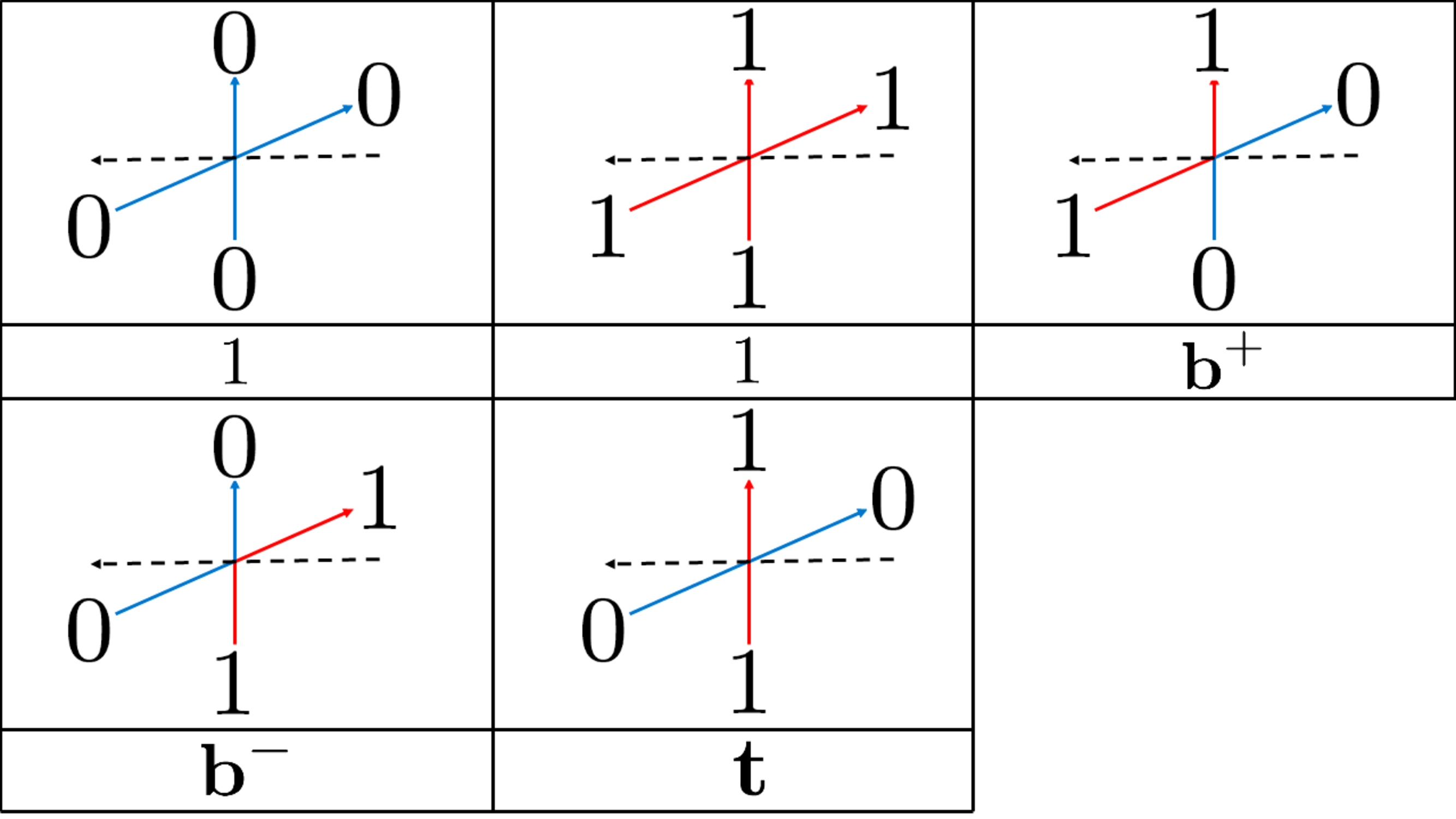}
\caption{The operator-valued matrix elements of $\mathcal{R}$.
We color the half-edges which are connected to 1s in the two-dimensional vector spaces with
red, and color 0s with
blue.
}
\label{3DRtwofigure}
\end{figure}

We denote the basis of $\mathcal{F}_{k\ell}$
by $\{ |m_{k \ell} \rangle_{k \ell} \;;\; (  m_{k \ell}=0,1,2,\dots ) \}$.
The standard bases of $\mathcal{F}^{\otimes n(n-1)/2}$
and its dual $(\mathcal{F}^{\otimes n(n-1)/2})^*$ are given by
$\{
\otimes_{ (k,\ell)\in D_n} |m_{k \ell} \rangle_{k \ell} \;;\;
(m_{k \ell}=0,1,2,\dots ) 
\}
$ 
and 
$\{
\otimes_{(k,\ell) \in D_n}  \ {}_{k \ell} \langle m_{k \ell} | \;;\;
(m_{k \ell}=0,1,2,\dots ) 
\}
$, respectively.
We introduce notations for the vacuum state and its dual as
$|\Omega \rangle:=|0 \rangle^{\otimes n(n-1)/2} \in \mathcal{F}^{\otimes n(n-1)/2}$ and 
$\langle \Omega|:=(|\Omega \rangle)^*
\in (\mathcal{F}^{\otimes n(n-1)/2})^*$.

We introduce the linear operator $X_i^{(n)}(z)$ $(i=0,\dots,n)$ \cite{Kuniba,KMO1,KMO2} acting on the tensor space $\mathcal{F}^{\otimes n(n-1)/2}\otimes \mathbb{C}[z]$, which is graphically represented as Figure \ref{XIoperatorfigure}. 
Here, to each of $\frac{n(n-1)}{2}$ intersections,
a matrix element of $\mathcal{R}$ is assigned, which is determined according to the configuration of the red and blue segments.
$X_i^{(n)}(z)$ is a weighted sum (see Figure \ref{XIoperatorfigure}) of all matrix elements of $\mathcal{R}^{\otimes n(n-1)/2}$ determined according to the possible configuration in which (i) the segments at the bottom of $1,2,\dots,i$-th columns are colored by red, and
(ii) the segments at the bottom of $i+1,i+2,\dots,n$-th columns are colored by blue.
Graphically, the operator $X^{(n)}_i(z)$ sends a vector in $\mathcal{F}^{\otimes n(n-1)/2}\otimes \mathbb{C}[z]$ to the next slice.
We will abbreviate $X_i^{(n)}(z)$ as $X_i(z)$ if there is no confusion.

\begin{figure}[htbp]
\centering
\includegraphics[width=12truecm]{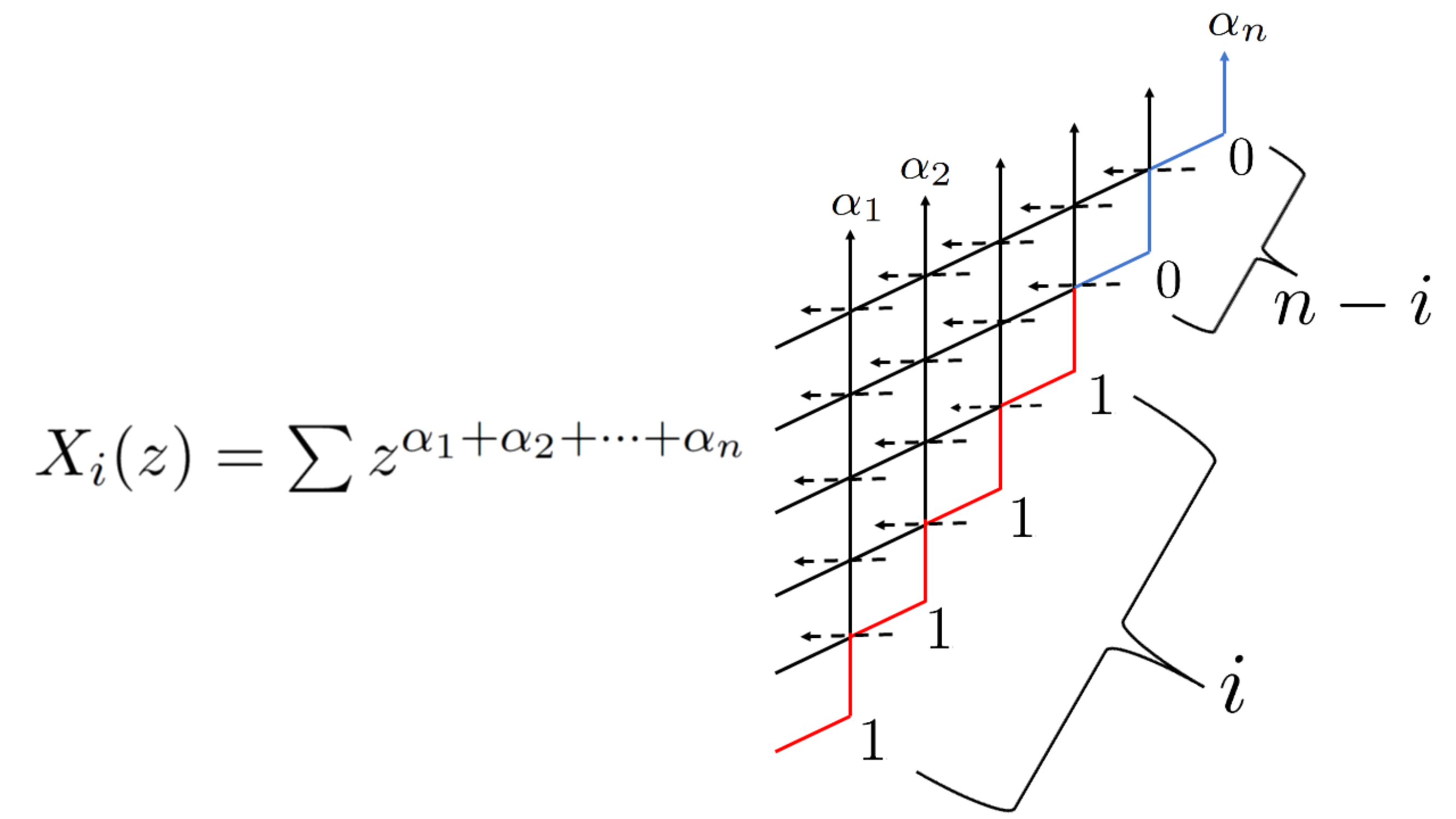}
\caption{The operator $X_i(z)=X_i^{(n)}(z)$.
We take the sum over all configurations except one side of the boundary
is fixed. The fixed boundary condition on one side is labeled by the subscript $i$ of the operator,
which corresponds to the number of consecutive number of integers which is fixed as $1$
on the fixed boundary. The sum is weighted, and the power $\alpha_1+\alpha_2+\cdots+\alpha_n$ of $z$
is the number of $1$s at the top boundary.
}
\label{XIoperatorfigure}
\end{figure}

In this paper, we investigate the following three-dimensional partition function
\begin{align}\label{eq:partition_function}
\langle \Omega |X_{i_1}(z_{1}) X_{i_{2}} (z_{2}) \cdots
X_{i_{m-1}}(z_{m-1}) X_{i_m}(z_{m}) | \Omega \rangle,
\end{align}
which is graphically represented as Figure \ref{3Dpartitionfunctionfigure}.
\begin{figure}[htbp]
\centering
\includegraphics[width=12truecm]{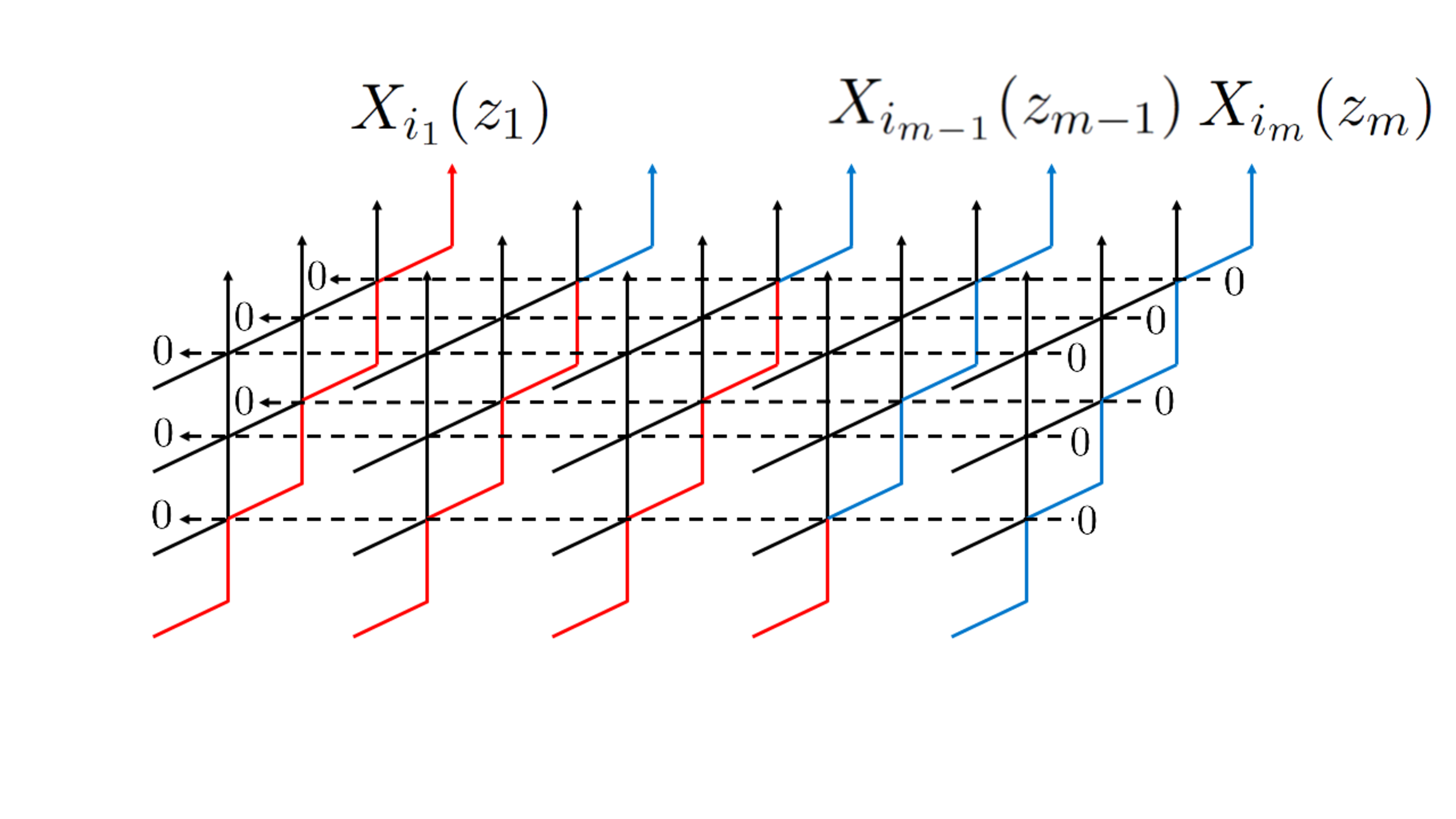}
\caption{The three-dimensional partition functions
$\langle \Omega| X_{i_1}(z_{1}) X_{i_{2}} (z_{2}) \cdots
X_{i_{m-1}}(z_{m-1}) X_{i_m}(z_{m})|\Omega \rangle$.
}
\label{3Dpartitionfunctionfigure}
\end{figure}

We mainly investigate the case $n \geq i_1 \geq i_{2} \geq \dots \geq i_{m-1} \geq i_m \geq 0$ in this paper.

\section{Three-dimensional partition functions and Schur functions}
\label{sec:Schur}
In this section, we study the partition function \eqref{eq:partition_function}.
Let $\Lambda_i$ be the set of all possible configurations of the red and blue segments that fit Figure \ref{XIoperatorfigure}.
We expand $X_i(z)$ as 
\begin{equation}\label{eq:expansion_of_Xi}
X_i(z)=\sum_{p\in \Lambda_i} z^{\alpha(p)} \bigotimes_{(k,\ell)\in D_n}x_{k \ell}(p),    
\end{equation}
where $\alpha(p)$ is the number of $1$s in the topmost row in $p$ and
$x_{k \ell}(p):\mathcal{F}_{k \ell}\to \mathcal{F}_{k \ell}$.
Then, $x_{k \ell}(p)$ is either $1$, $\mathbf{b}^+$, $\mathbf{b}^-$, or $\mathbf{t}$.
We sometimes use the notation $x_{k \ell}(p)=1_b$ (\textit{resp.}~$x_{k \ell}(p)=1_r$) if $x_{k \ell}(p)=1$ and it arises from two blue (\textit{resp.}~red) lines crossing at $(k, \ell)$.

The following is an immediate consequence of Figure \ref{3DRtwofigure}.
\begin{lemma}\label{lemma:conditioned_sum}
We have the following:
\begin{enumerate}
\item If $x_{k+1,\ell}(p)=1_b$ and $x_{k,\ell+1}(p)=1_b$, then $x_{k,\ell}(p)=1_b$ or $\mathbf{b}^+$.
\item If $x_{k+1,\ell}(p)=1_r$ and $x_{k,\ell+1}(p)=1_r$, then $x_{k,\ell}(p)=1_r$ or $\mathbf{b}^-$.
\item If $x_{k+1,\ell}(p)=\mathbf{t}$ and $x_{k,\ell+1}(p)=1_b$, then $x_{k,\ell}(p)=\mathbf{t}$.
\item If $x_{k+1,\ell}(p)=1_r$ and $x_{k,\ell+1}(p)=\mathbf{t}$, then $x_{k,\ell}(p)=\mathbf{t}$.
\end{enumerate}
\end{lemma}

We first show the following factorized expression
for cases $n \geq i_m \geq i_{m-1} \geq \dots \geq i_2 \geq i_1 \geq 0$.
\begin{lemma} \label{simplestpartition}
When $n \geq i_m \geq i_{m-1} \geq \dots \geq i_2 \geq i_1 \geq 0$,
we have
\begin{align}
\langle \Omega| X_{i_1}(z_{1}) X_{i_{2}} (z_{2}) \cdots
X_{i_{m-1}}(z_{m-1}) X_{i_m}(z_{m})|\Omega \rangle
=z_1^{i_1} \cdots z_{m}^{i_m}.
\label{top}
\end{align}
\end{lemma}
\begin{proof}
We prove the lemma
by showing that there is only one 
configuration that contributes to \eqref{top}.
This can be checked by investigating the bosonic Fock spaces $\mathcal{F}_{k\ell}$ assigned to
$(k,\ell)=(n-1,1) \to (n-2,2) \to (n-2,1) \to (n-3,3) \to (n-3,2) \to (n-3,1) \to \cdots \to (1,1)$.
We first see what happens to the Fock space $\mathcal{F}_{(n-1,1)}$.
Let $\mathfrak{X}:=X_{i_1}(z_{1}) X_{i_{2}} (z_{2}) \cdots
X_{i_{m-1}}(z_{m-1}) X_{i_m}(z_{m})$.
Let $\theta_k$ denote the number of $p$ such that $i_p=k$.
We let $\theta_{\leq k}$ denote the sum $\theta_1+\theta_2+\dots+\theta_k$.
In a similar manner, we define $\theta_{<k}$, $\theta_{\geq k}$, and $\theta_{> k}$.
By definition of $X_0(z)$, the $\mathcal{F}_{(n-1,1)}$ component of each term in $X_0(z)$ must be $1_b$ or $\mathbf{b}^+$.
From Figure \ref{3DRtwofigure}, we also find that the $\mathcal{F}_{(n-1,1)}$ component of a term in $X_1(z)$ is $\mathbf{t}$, and that in $X_p(z)$ $(p>1)$ is $1_r$ or $\mathbf{b}^-$.
Hence, the $\mathcal{F}_{(n-1,1)}$ component of each term in $\mathfrak{X}$ should be a sequence of the form $(x_1\dots x_{\theta_0})\mathbf{t}^{\theta_1}(y_1\dots y_{\theta_{\geq 2}})$ such that $x_j\in \{1_b,\mathbf{b}^+\}$ and $y_j\in \{1_r,\mathbf{b}^-\}$.
However, since $\langle 0 |\mathbf{b}^+=0$ and $\mathbf{b}^-|0\rangle=0$, terms that can survive should satisfy $x_j=1_b$ and $y_j=1_r$ for all $j$. 
By the same reason, the $\mathcal{F}_{(n-2,2)}$ components of terms in $\mathfrak{X}$ that can survive should be of the form $1_b^{\theta_{\leq 1}}\mathbf{t}^{\theta_2} 1_r^{\theta_{\geq 3}}$.
Therefore, from Lemma \ref{lemma:conditioned_sum}, we find that the $\mathcal{F}_{(n-2,1)}$ components of the survived terms are of the form $(x_1\dots x_{\theta_0})\mathbf{t}^{\theta_1 + \theta_2} 
(y_1\dots y_{\theta_{> 2}})$.
For the same reason we previously noted, terms that can survive should be of the form $1_b^{\theta_0}\mathbf{t}^{\theta_1+\theta_2} 
1_r^{\theta_{> 2}}$.
We can repeat this procedure and show that the $\mathcal{F}_{kl}$ components of terms that can survive should be of the form $1_b^{\theta_{<l}}\mathbf{t}^{\theta_l+ \cdots +\theta_{n-k}}
1_r^{\theta_{>n-k}}$.
Therefore, all the possible configurations at each intersection must consist of (i) two blue lines, (ii) two red lines, or (iii) one vertical red line and one horizontal blue line.
Figure \ref{XIoperatorfreezefigure} presents the only possible configuration on a slice corresponding to $X_i(z)$, the contribution of which is $z^i$.
Multiplying these factors gives \eqref{top}.
\end{proof}

\begin{figure}[htbp]
\centering
\includegraphics[width=12truecm]{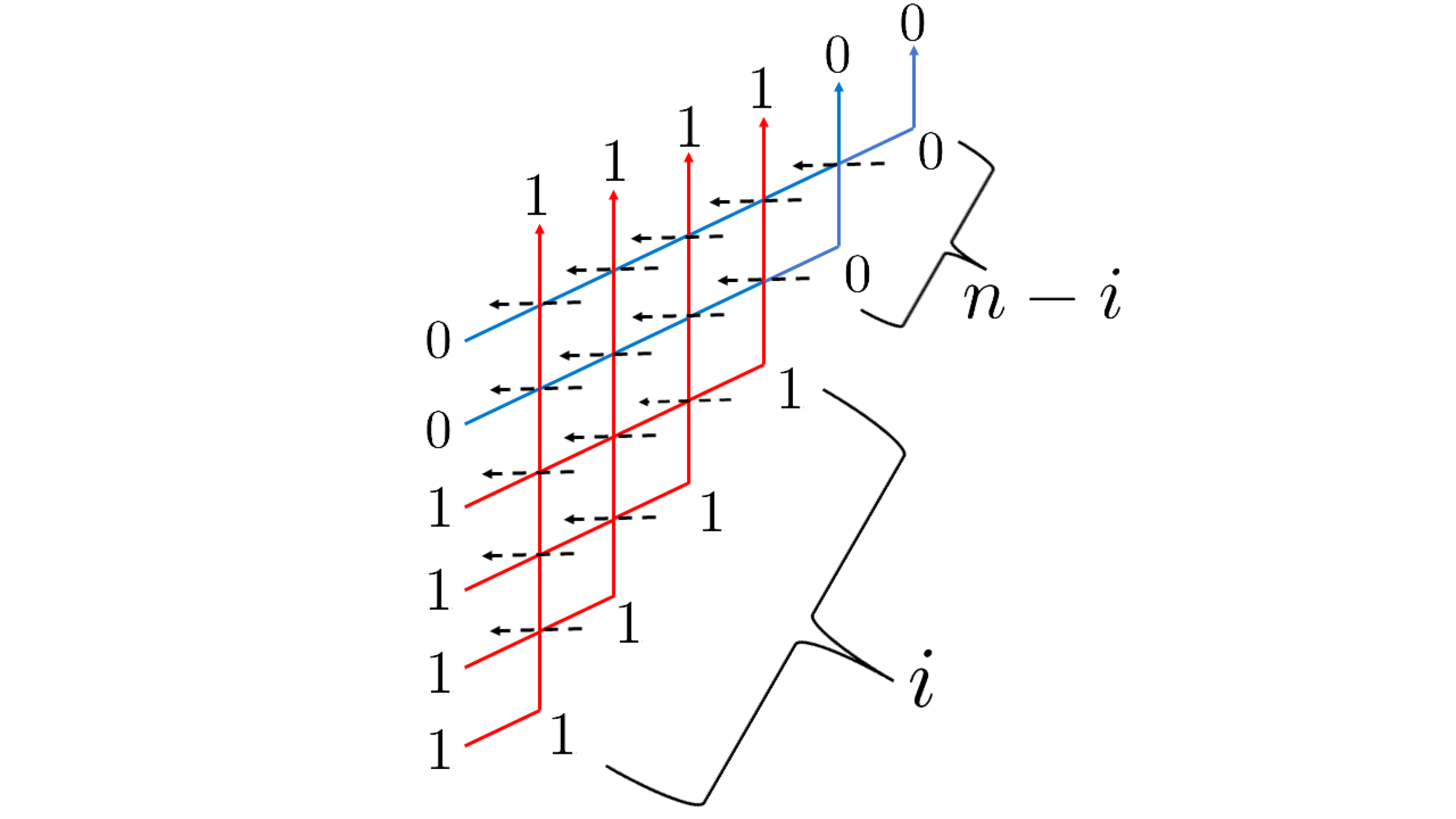}
\caption{The unique configuration for the operator $X_i(z)$ in
$\langle \Omega| X_{i_1}(z_{1}) X_{i_{2}} (z_{2}) \cdots
X_{i_{m-1}}(z_{m-1}) X_{i_m}(z_{m})|\Omega \rangle$ when $n \geq i_m \geq i_{m-1} \geq \dots \geq i_2 \geq i_1 \geq 0$.
One notes from this configuration that from $X_i(z)$ the factor $z^i$ contributes.
}
\label{XIoperatorfreezefigure}
\end{figure}

The operators $X_i(z)$
are shown to satisfy the following relations.
\begin{theorem} \cite{Kuniba,KMO2} \label{ThmFZ}
The operators $X_i(z)$ satisfy the  Zamolodchikov-Faddeev algebra relations
\begin{align}
X_i(x)X_j(y)=\left\{
\begin{array}{ll}
X_i(y)X_j(x)+(1-x/y)X_j(y)X_i(x) & i<j, \\
X_i(y)X_i(x) & i=j, \\
x/y X_i(y) X_j(x) & i>j,
\end{array}
\right.
.
\end{align} \label{FZalg}
\end{theorem}
Theorem \ref{ThmFZ} is proved in \cite{Kuniba,KMO2} as a consequence
of the tetrahedron equation
(see Appendix A for the tetrahedron equation).

Recall the Schur polynomials.
For a sequence of integers $\lambda=(\lambda_1,\lambda_2,\dots,\lambda_n)$
where
$\lambda_1 \geq \lambda_2 \geq \cdots \geq \lambda_n \geq 0$,
the Schur polynomials $s_{\lambda}(z_1,\dots,z_n)$ is
\begin{align}
s_\lambda(z_1,\dots,z_n)=\frac{\det_{1 \le i,j \le n } (z_i^{\lambda_j+n-j})}{\prod_{1 \le i <j \le n}(z_i-z_j) }.
\end{align}

Let us introduce several shorthand notations.
For a set of variables ${\bf z}=\{z_1,\dots,z_{|{\bf z}|} \}$
denote $X_i({\bf z}):=\prod_{j=1}^{|{\bf z}|} X_i(z_j)$.
Note that due to commutativity $X_i(z)X_i(w)=X_i(w)X_i(z)$, the ordering of variables in the product
does not matter. For an arbitrary integer $p$,
define ${\bf z}^p:=\prod_{j=1}^{|{\bf z}|} z_j^p$.
We also introduce the following notations ${\bf z}-{\bf w}:=\prod_{j=1}^{|{\bf z}|} \prod_{k=1}^{|{\bf w}|} (z_j-w_k)$,
$1-{\bf z}/{\bf w}:=\prod_{j=1}^{|{\bf z}|} \prod_{k=1}^{|{\bf w}|}
(1-z_j/w_k)$
for two sets of variables ${\bf z},{\bf w}$.

We use the following formula for Schur functions
($t=0$ case of \cite{Pragacz}, Proposition 4).
\begin{proposition} \cite{Pragacz} \label{Symmformula}
For ${\bf z}=({\bf z}_1,{\bf z}_2,\dots,{\bf z}_m)$, we have
\begin{align}
s_{(\lambda_1^{|{\bf z}_1|},\lambda_2^{|{\bf z}_2|},\dots,\lambda_m^{|{\bf z}_m|})}({\bf z})=\sum_{{\bf w}=({\bf w}_1,{\bf w}_2,\dots,{\bf w}_m)}
\frac{1}{\prod_{1 \le j<k \le m}({\bf w}_j-{\bf w}_k) } {\bf w}_1^{\lambda_1+m-1} {\bf w}_2^{\lambda_2+m-2}
\cdots {\bf w}_m^{\lambda_m}.
\label{symmformulaSchur}
\end{align}
Here, we take the sum over all 
${\bf w}=({\bf w}_1,{\bf w}_2,\dots,{\bf w}_m) $
such that
${\bf w}_i \ (i=1,\dots,m)$ are unordered sets of variables satisfying $|{\bf w}_i|=|{\bf z}_i|$
and 
${\bf w}_1  \cup {\bf w}_2  \cup \cdots \cup {\bf w}_m={\bf z}_1 \cup {\bf z}_2  \cup \cdots \cup {\bf z}_m$
.
\end{proposition}

We prepare the following multiple commutation relations
between the operators $X_i(z)$.
\begin{proposition} \label{multipleformula}
For $n \geq i_1 > i_{2} > \cdots > i_m \geq 0$,
we have the following commutation relations:
\begin{align}
&\frac{X_{i_1}({\bf z}_1) X_{{i_{2}}}({\bf z}_{2}) \cdots X_{i_{m-1}}({\bf z}_{m-1}) X_{i_m}({\bf z}_m)}
{\prod_{k=1}^m {\bf z}_k^{m-k}}
\nonumber \\
=&\sum_{{\bf w}=({\bf w}_1,{\bf w}_2,\dots,{\bf w}_m)}
\frac{1}{\prod_{1 \le j<k \le m}({\bf w}_j-{\bf w}_k) } 
 X_{i_m}({\bf w}_m) X_{i_{m-1}}({\bf w}_{m-1}) \cdots X_{i_2}({\bf w}_{2}) X_{i_1}({\bf w}_1).
\label{multiplecommrel}
\end{align}
Here, we take the sum over all 
${\bf w}=({\bf w}_1,{\bf w}_2,\dots,{\bf w}_m) $
such that
${\bf w}_i \ (i=1,\dots,m)$ are unordered sets of variables satisfying $|{\bf w}_i|=|{\bf z}_i|$
and 
${\bf w}_1  \cup {\bf w}_2  \cup \cdots \cup {\bf w}_m={\bf z}_1 \cup {\bf z}_2  \cup \cdots \cup {\bf z}_m$
.
\end{proposition}

\begin{proof}
We can give a proof using the argument given in \cite{Motegi,SU}.
We rewrite the Zamolodchikov-Faddeev algebra relations \eqref{FZalg}
in the following forms
\begin{align}
X_i(x)X_j(y)&=(1-y/x)^{-1} X_j(y) X_i(x)-(1-y/x)^{-1} X_j(x) X_i(y), \ i>j, \label{commrelone} \\
X_i(x)X_i(y)&=X_i(y)X_i(x), \label{commreltwo} \\
X_i(x)X_j(y)&=x/y X_i(y) X_j(x), \  i>j. \label{commrelthree}
\end{align}
From these commutation relations, one notes there exist commutation relations
of the following form
\begin{align}
&\frac{X_{i_1}({\bf z}_1) X_{{i_{2}}}({\bf z}_{2}) \cdots X_{i_{m-1}}({\bf z}_{m-1}) X_{i_m}({\bf z}_m)}
{\prod_{k=1}^m {\bf z}_k^{m-k}}
\nonumber \\
=&
\sum_{{\bf w}=({\bf w}_1,{\bf w}_2,\dots,{\bf w}_m)}
A({\bf w}_1,{\bf w}_2,\dots,{\bf w}_m) 
X_{i_m}({\bf w}_m) X_{i_{m-1}}({\bf w}_{m-1}) \cdots X_{i_2}({\bf w}_{2}) X_{i_1}({\bf w}_1),
\label{firstweknow}
\end{align}
and the problem is to determine $A({\bf w}_1,{\bf w}_2,\dots,{\bf w}_m) $ explicitly.
Fix ${\bf w}=({\bf w}_1,{\bf w}_2,\dots,{\bf w}_m) $
and first note that the left hand side of \eqref{firstweknow}
can be rewritten using \eqref{commreltwo}, \eqref{commrelthree} as
\begin{align}
&\frac{X_{i_1}({\bf z}_1) X_{{i_{2}}}({\bf z}_{2}) \cdots X_{i_{m-1}}({\bf z}_{m-1}) X_{i_m}({\bf z}_m)}
{\prod_{k=1}^m {\bf z}_k^{m-k}}
\nonumber \\
=&\frac{X_{i_1}({\bf w}_1) X_{{i_{2}}}({\bf w}_{2}) \cdots X_{i_{m-1}}({\bf w}_{m-1}) X_{i_m}({\bf w}_m)}
{\prod_{k=1}^m {\bf w}_k^{m-k}}. \label{rearrange}
\end{align}
Next, we reverse the order of the operators in
the right hand side of \eqref{rearrange}
using \eqref{commrelone}.
We note that to get the term
$X_{i_m}({\bf w}_m) X_{i_{m-1}}({\bf w}_{m-1}) \cdots X_{i_2}({\bf w}_{2}) X_{i_1}({\bf w}_1)$ in this way,
we only need to keep track of the first term of the right hand side of \eqref{commrelone},
and we get the factor $\frac{1}{\prod_{1 \le j < k \le m} (1-{\bf w_k}/{\bf w_j})}$.
Hence we find $A({\bf w}_1,{\bf w}_2,\dots,{\bf w}_m)$ is
\begin{align}
A({\bf w}_1,{\bf w}_2,\dots,{\bf w}_m)
=\frac{1}{\prod_{k=1}^m {\bf w}_k^{m-k}}
\frac{1}{\prod_{1 \le j < k \le m} (1-{\bf w_k}/{\bf w_j})}
=\frac{1}{\prod_{1 \le j<k \le m}({\bf w}_j-{\bf w}_k) }.
\end{align}
These coefficients are uniquley determined by 
Lemma \ref{simplestpartition}.
\end{proof}

Using Lemma \ref{simplestpartition},
Proposition \ref{Symmformula}, Proposition \ref{multipleformula},
we determine the explicit forms of the partition functions.

\begin{theorem} \label{mainthm}
For $n \geq i_1 > i_{2} > \cdots > i_m \geq 0$, we have
\begin{align}
&\langle \Omega |
X_{i_1}({\bf z}_1) X_{{i_{2}}}({\bf z}_{2}) \cdots X_{i_{m-1}}({\bf z}_{m-1}) X_{i_m}({\bf z}_m)
| \Omega \rangle
\nn \\
=&\prod_{k=1}^m {\bf z}_k^{m-k}
s_{((i_1-m+1)^{|{\bf z}_1|},(i_2-m+2)^{|{\bf z}_2|},\dots,i_m^{|{\bf z}_m|})}({\bf z}).
\end{align}
\end{theorem}

In particular, the simplest case $|{\bf z}_1|=\cdots=|{\bf z}_m|=1$ gives the following.
\begin{theorem} 
For $n \geq i_1 > i_{2} > \cdots > i_m \geq 0$, we have
\begin{align}
&\langle \Omega | X_{i_1}(z_1) X_{i_{2}}(z_{2}) \cdots X_{i_m}(z_{m}) | \Omega \rangle
\nn \\
=&\prod_{k=1}^{m} z_k^{m-k} s_{(i_1-m+1,i_{2}-m+2,\dots,i_m)}(z_1,\cdots,z_m). \label{casenomult}
\end{align}
\end{theorem}

\begin{proof}
From \eqref{multiplecommrel},
we have the following relation between partition functions
\begin{align}
&\langle \Omega| X_{i_1}({\bf z}_1) X_{{i_{2}}}({\bf z}_{2}) \cdots X_{i_{m-1}}({\bf z}_{m-1}) X_{i_m}({\bf z}_m) |\Omega \rangle
\nonumber \\
=&\prod_{k=1}^m {\bf z}_k^{m-k}
\sum_{{\bf w}=({\bf w}_1,{\bf w}_2,\dots,{\bf w}_m)}
\frac{1}{\prod_{1 \le j<k \le m}({\bf w}_j-{\bf w}_k) } \nn \\
&\times
\langle \Omega| X_{i_m}({\bf w}_m) X_{i_{m-1}}({\bf w}_{m-1}) \cdots X_{i_2}({\bf w}_{2}) X_{i_1}({\bf w}_1) | \Omega \rangle.
\label{auxproofSchurone}
\end{align}
For $n \geq i_1 > i_{2} > \cdots > i_m \geq 0$, we have
\begin{align}
\langle \Omega| X_{i_m}({\bf w}_m) X_{i_{m-1}}({\bf w}_{m-1}) \cdots X_{i_2}({\bf w}_{2}) X_{i_1}({\bf w}_1) | \Omega \rangle={\bf w}_1^{i_1}
{\bf w}_2^{i_2}
\cdots {\bf w}_{m-1}^{i_{m-1}} {\bf w}_m^{i_m}.
\label{auxproofSchurtwo}
\end{align}
by Lemma \ref{simplestpartition}.
Combining \eqref{auxproofSchurone} and \eqref{auxproofSchurtwo}
we have
\begin{align}
&\langle \Omega| X_{i_1}({\bf z}_1) X_{{i_{2}}}({\bf z}_{2}) \cdots X_{i_{m-1}}({\bf z}_{m-1}) X_{i_m}({\bf z}_m) |\Omega \rangle
\nonumber \\
=&\prod_{k=1}^m {\bf z}_k^{m-k}
\sum_{{\bf w}=({\bf w}_1,{\bf w}_2,\dots,{\bf w}_m)}
\frac{1}{\prod_{1 \le j<k \le m}({\bf w}_j-{\bf w}_k) }
{\bf w}_1^{i_1}
{\bf w}_2^{i_2}
\cdots {\bf w}_{m-1}^{i_{m-1}} {\bf w}_m^{i_m}.
\label{auxproofSchurthree}
\end{align}
Using \eqref{symmformulaSchur}, we note the right hand side of
\eqref{auxproofSchurthree} can be rewritten as
\begin{align}
\prod_{k=1}^m {\bf z}_k^{m-k}
s_{((i_1-m+1)^{|{\bf z}_1|},(i_2-m+2)^{|{\bf z}_2|},\dots,i_m^{|{\bf z}_m|})}({\bf z}).
\end{align}

\end{proof}

Let us introduce hat operators
\cite{KMO2} and natural generalizations
\begin{align}
\hat{X}_i^k(z):=\frac{d^k}{d z^k} X_i(z), \ \ \ k=0,1,\dots.
\end{align}

From \eqref{casenomult}, we have the following.

\begin{corollary} 
For $n \geq i_1 > i_{2} > \cdots > i_m \geq 0$, we have
\begin{align}
&\langle \Omega | \hat{X}_{i_1}^{j_1}(z_1) \hat{X}_{i_{2}}^{j_2}(z_{2}) \cdots
\hat{X}_{i_m}^{j_m}(z_{m}) | \Omega \rangle
\nn \\
=&\frac{\partial^{j_1}}{\partial z_1^{j_1}}
\frac{\partial^{j_2}}{\partial z_2^{j_2}}
\cdots
\frac{\partial^{j_m}}{\partial z_m^{j_m}}
\Bigg(
\prod_{k=1}^{m} z_k^{m-k} 
s_{(i_1-m+1,i_{2}-m+2,\dots,i_m)}(z_1,\cdots,z_m)
\Bigg). 
\end{align}
\end{corollary}

Let us see some examples. Consider the case $m=n$, $i_j=n-j+1$ $(j=1,\dots,\ell)$,
$i_j=n-j$ $(j=\ell+1,\dots,n)$.
From \eqref{casenomult}, we have
\begin{align}
&\langle \Omega | X_{n}(z_1)  \cdots X_{n-\ell+1}(z_\ell)
X_{n-\ell-1}(z_{\ell+1}) \cdots
X_{0}(z_{n}) | \Omega \rangle
\nn \\
=&\prod_{k=1}^{n} z_k^{n-k} s_{(1^\ell,0^{n-\ell})}(z_1,\dots,z_n)
=\prod_{k=1}^{n} z_k^{n-k}
e_\ell(z_1,\dots,z_n),
\end{align}
where $e_\ell(z_1,\dots,z_n)$
is the elementary symmetric functions \\
$e_\ell(z_1,\dots,z_n)=\sum_{1 \le k_1 < k_2 < \cdots < k_\ell \le n}
z_{k_1} z_{k_2} \cdots z_{k_\ell}$.
Using 
\begin{align}
    \frac{\partial}{\partial z_k}
    e_\ell(z_1,\dots,z_n)
    =e_{\ell-1}(z_1,\dots,z_{k-1},z_{k+1},\dots,z_n),
    \label{differentialelementary}
\end{align}
we have
\begin{align}
    &\langle \Omega | \hat{X}_{n}(z_1) X_{n-1}(z_2) \cdots X_{n-\ell+1}(z_\ell)
X_{n-\ell-1}(z_{\ell+1}) \cdots
X_{0}(z_{n}) | \Omega \rangle
\nn \\
=&(n-1) z_1^{n-2} \prod_{k=2}^{n} z_k^{n-k}
e_\ell(z_1,\dots,z_n)
+\prod_{k=1}^{n} z_k^{n-k} e_{\ell-1}(z_2,\dots,z_n).
\label{examplehat}
\end{align}

More generally, first
recall the Jacobi-Trudi formula for the Schur functions
\begin{align}
s_\lambda(z_1,\dots,z_n)
=\det_{1 \le i,j \le \ell(\lambda^\prime)}(e_{\lambda_i^\prime-i+j}(z_1,\dots,z_n)), \label{JTformula}
\end{align}
where $\lambda_i^\prime=(\lambda_1^\prime,\lambda_2^\prime,\dots,\lambda_{\ell(\lambda^\prime)}^\prime)$ is the conjugate partition of $\lambda$,
and
$\ell(\lambda^\prime)$ is the length of $\lambda^\prime$.
Using \eqref{casenomult}, \eqref{differentialelementary},
\eqref{JTformula} and the property of differential operators
acting on determinants, we have
\begin{align}
&\frac{\partial}{\partial z_k}
s_\lambda(z_1,\dots,z_n) \nn \\
=&\sum_{\ell=1}^{\ell(\lambda^\prime)}
\det_{1 \le i,j \le \ell(\lambda^\prime)}
\Bigg( (1-\delta_{j,\ell}
+\delta_{j,\ell} \frac{\partial}{\partial z_k}
)
e_{\lambda_i^\prime-i+j}(z_1,\dots,z_n)
\Bigg) \nn \\
=&\sum_{\ell=1}^{\ell(\lambda^\prime)}
\det_{1 \le i,j \le \ell(\lambda^\prime)}
\Bigg( (1-\delta_{j,\ell})
e_{\lambda_i^\prime-i+j}(z_1,\dots,z_n)
+\delta_{j,\ell}
e_{\lambda_i^\prime-i+j-1}(z_1,\dots,z_{k-1},z_{k+1},\dots,z_n)
\Bigg).
\end{align}
Hence we get
\begin{align}
&\langle \Omega | \hat{X}_{\lambda_1+m-1}(z_1) X_{\lambda_{2}+m-2}(z_{2}) \cdots
X_{\lambda_m}(z_{m}) | \Omega \rangle
\nn \\
=&\frac{\partial}{\partial z_1}
\Bigg(
\prod_{k=1}^{m} z_k^{m-k} 
s_{(\lambda_1,\lambda_2,\dots,\lambda_m)}(z_1,\cdots,z_m)
\Bigg) \nn \\
=&(m-1)z_1^{m-2} \prod_{k=2}^{m} z_k^{m-k}
s_{(\lambda_1,\lambda_2,\dots,\lambda_m)}(z_1,\cdots,z_m)
+\prod_{k=1}^{m} z_k^{m-k} \nn \\
&\times \sum_{\ell=1}^{\ell(\lambda^\prime)}
\det_{1 \le i,j \le \ell(\lambda^\prime)}
\Bigg( (1-\delta_{j,\ell})
e_{\lambda_i^\prime-i+j}(z_1,\dots,z_m)
+\delta_{j,\ell}
e_{\lambda_i^\prime-i+j-1}(z_1,\dots,z_{k-1},z_{k+1},\dots,z_m)
\Bigg)
. \label{derSchur}
\end{align}

Let us also see the case when the variables are all specialized to 1.
Recall the factorization of Schur polynomials when specializing the variables
$z_1=z_2=\dots=z_n=1$
\begin{align}
s_\lambda(1^n)=\prod_{1 \le k < \ell  \le n} \frac{\lambda_k-\lambda_\ell+\ell-k}{\ell-k},
\end{align}
which is well-known, which follows from the determinant form for example.
Applying this factorization to \eqref{casenomult}, we have the following.
\begin{corollary}
For $n \geq i_1 > i_{2} > \cdots > i_m \geq 0$, we have
\begin{align}
&\langle \Omega | X_{i_1}(1) X_{i_{2}}(1) \cdots X_{i_m}(1) | \Omega \rangle
=\prod_{1 \le k < \ell \le m} \frac{i_k-i_\ell}{\ell-k}. 
\label{countingconfigs}
\end{align}
\end{corollary}
This corollary corresponds to counting the number of configurations
contributing to the partition functions.

From
\eqref{derSchur} and
\eqref{countingconfigs}
we have the following.
\begin{corollary}
We have
\begin{align}
&\frac{
\langle \Omega | \hat{X}_{\lambda_1+m-1}(1) X_{\lambda_{2}+m-2}(1) \cdots
X_{\lambda_m}(1) | \Omega \rangle}
{\langle \Omega | X_{\lambda_1+m-1}(1) X_{\lambda_{2}+m-2}(1) \cdots
X_{\lambda_m}(1) | \Omega \rangle}
=m-1
\nn \\
&
+\prod_{1 \le k < \ell  \le m} \frac{\ell-k}{\lambda_k-\lambda_\ell+\ell-k}
\sum_{\ell=1}^{\ell(\lambda^\prime)}
\det_{1 \le i,j \le \ell(\lambda^\prime)}
\Bigg( (1-\delta_{j,\ell})
\binom{m}{\lambda_i^\prime-i+j}
+\delta_{j,\ell}
\binom{m}{\lambda_i^\prime-i+j-1}
\Bigg)
. 
\end{align}
Here,
$\binom{n}{k}:=n!/k!(n-k)!$, $k!:=\prod_{j=1}^k j$.
\end{corollary}

The simplest cases are
\begin{align}
    &\frac{\langle \Omega | \hat{X}_{n}(1) X_{n-1}(1) \cdots X_{n-\ell+1}(1)
X_{n-\ell-1}(1) \cdots
X_{0}(1) | \Omega \rangle}
{\langle \Omega | X_{n}(1) X_{n-1}(1) \cdots X_{n-\ell+1}(1)
X_{n-\ell-1}(1) \cdots
X_{0}(1) | \Omega \rangle}=n-1+\frac{1}{\ell}.
\label{avnumber}
\end{align}
which can also be derived using \eqref{examplehat}.

Let us finally discuss an
application/interpretation of this result to
probability theory/statistical physics.
The shape for the partition functions
which we consider can be viewed as a prism
with two bases which are triangles and three lateral faces
which are rectangles.
To each  of the lateral faces, there are $mn$ points associated,
each point
colored with either red or blue.
To each edge of the triangles corresponding to bases,
there are $n$ points associated,
each of which is colored with red or blue.

The following quantity
\begin{align}
\frac{\langle \Omega | \hat{X}_{i_1}(1) X_{i_{2}}(1) \cdots
X_{i_m}(1) | \Omega \rangle}{\langle \Omega | X_{i_1}(1) X_{i_{2}}(1) \cdots
X_{i_m}(1) | \Omega \rangle},
\end{align}
corresponds to the average number of points lying on one edge of the base is colored with red, under the boundary condition
such that coloring on $mn$ points on
one of the lateral faces is fixed which can be labelled by
$(i_1,i_2,\dots,i_m)$,
and the boundaries corresponding to the other two lateral faces
are free.
\eqref{avnumber} means that for the case
when the fixed boundary is labelled by
$(i_1,i_2,\dots,i_{n})=(n,n-1,\dots,n-\ell+1,n-\ell-1,\dots,0)$,
the average number of points on one edge
of the base
colored with red is $\displaystyle n-1+\frac{1}{\ell}$.

\section{Inhomogeneous three-dimensional partition functions and loop elementary symmetric functions}
\label{sec:inhomo}

\begin{figure}[htbp]
\centering
\includegraphics[width=12truecm]{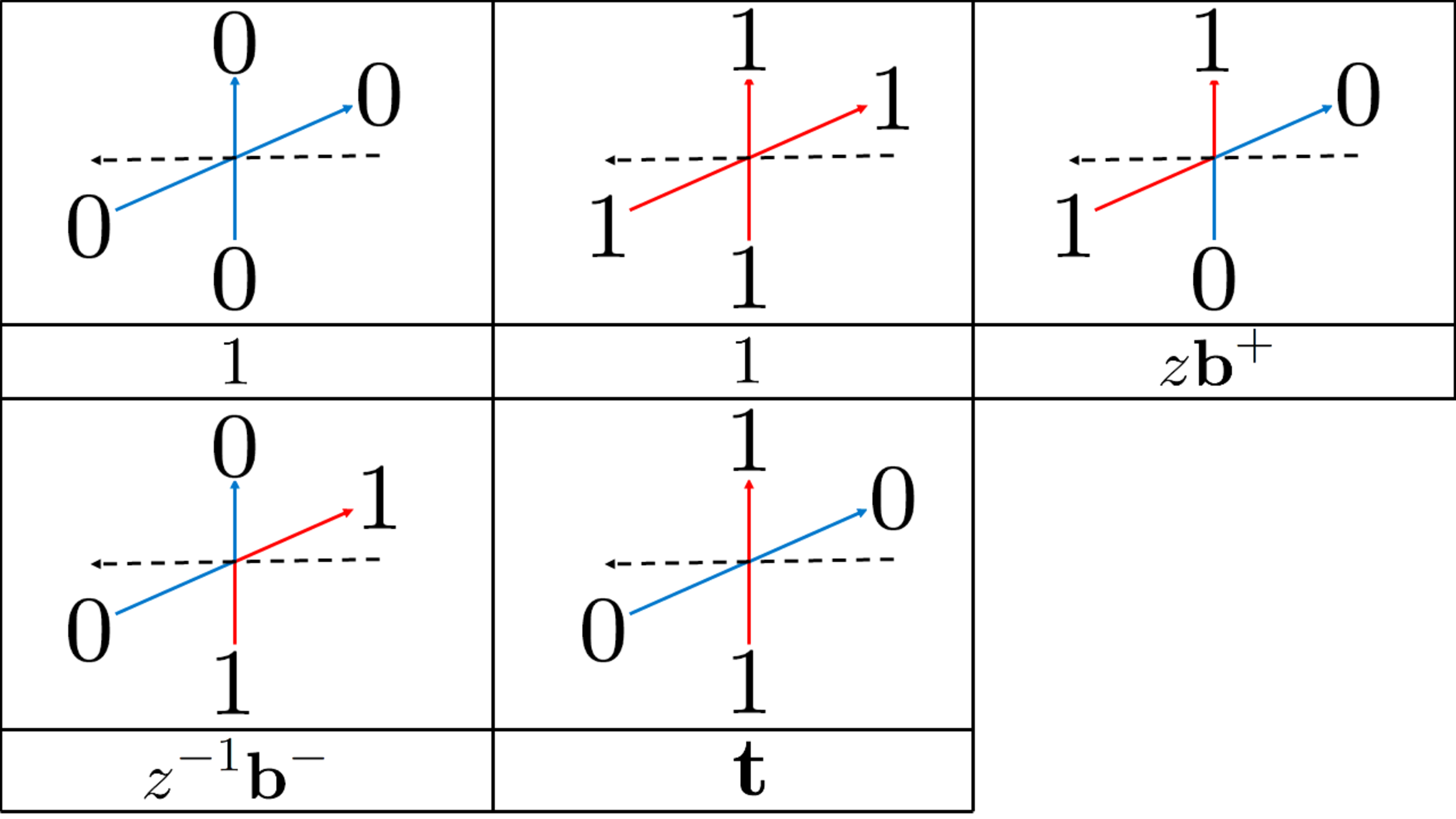}
\caption{The operator-valued matrix elements of $\mathcal{L}(z)$.
$\mathbf{b}^+$ and $\mathbf{b}^-$ in $\mathcal{R}$ is replaced by
$z \mathbf{b}^+$ and $z^{-1} \mathbf{b}^-$.
}
\label{inhomogeneous3DRfigure}
\end{figure}

This section discusses an inhomogeneous generalization
of the partition function.
We derive explicit representations for one of the simplest cases
by brute force computation.
We use the following three-dimensional $L$-operator $\mathcal{L}(z)$, which is a degeneration of the one in \cite{BaSe} (see \cite{Kuniba}, (18.23) for example)
\begin{align}
\mathcal{L}(z)( v_i \otimes v_j)
=\sum_{a=0}^1 \sum_{b=0}^1  v_a \otimes v_b [\mathcal{L}(z)]_{ij}^{ab},
\label{q=0L}
\end{align}
where
\begin{align}
[\mathcal{L}(z)]_{00}^{00}&=1, \ \ \
[\mathcal{L}(z)]_{11}^{11}=1, \ \ \
[\mathcal{L}(z)]_{10}^{01}=z \mathbf{b}^+, \\
[\mathcal{L}(z)]_{01}^{10}&=z^{-1} \mathbf{b}^-, \ \ \
[\mathcal{L}(z)]_{01}^{01}=\mathbf{t},
\end{align}
and $[\mathcal{L}(z)]_{ij}^{ab}=0$ otherwise.
Note $\mathcal{L}(1)=\mathcal{R}$.
See Figure \ref{inhomogeneous3DRfigure}.

Let $z=(z^{(k,\ell)})_{(k,\ell)\in D_n}$ be a set of independent variables.
Each $z^{(k,\ell)}$ is associated with the line $L_{k \ell}$.
We define the linear operator $\overline{X}_i(z)$, which acts on $\bigotimes_{(k,\ell)\in D_n}\mathcal{F}_{k \ell}\otimes \mathbb{C}[z^{(k,\ell)}]$, as the sum of matrix elements of $\bigotimes_{(k,\ell)\in D_n}\mathcal{L}(z^{(k,\ell)}) $
over all possible configurations that we mentioned in Section \ref{sec:three-simensional_Rmatrix}.
See Figure \ref{barfigure}.
Although $\overline{X}_i(z)$ depends on several variables
$z^{(k, \ell)}$ $((k,\ell)\in D_n)$,
we keep to use $z$.
Hereinafter, we refer to such $z$ as a \textit{multiple variable}.
When $z^{(k, \ell)}=z$ for all $(k,\ell)$,
$\overline{X}_i(z)$ reduces to $z^{-i} X_i(z)$.

\begin{figure}[htbp]
\centering
\includegraphics[width=12truecm]{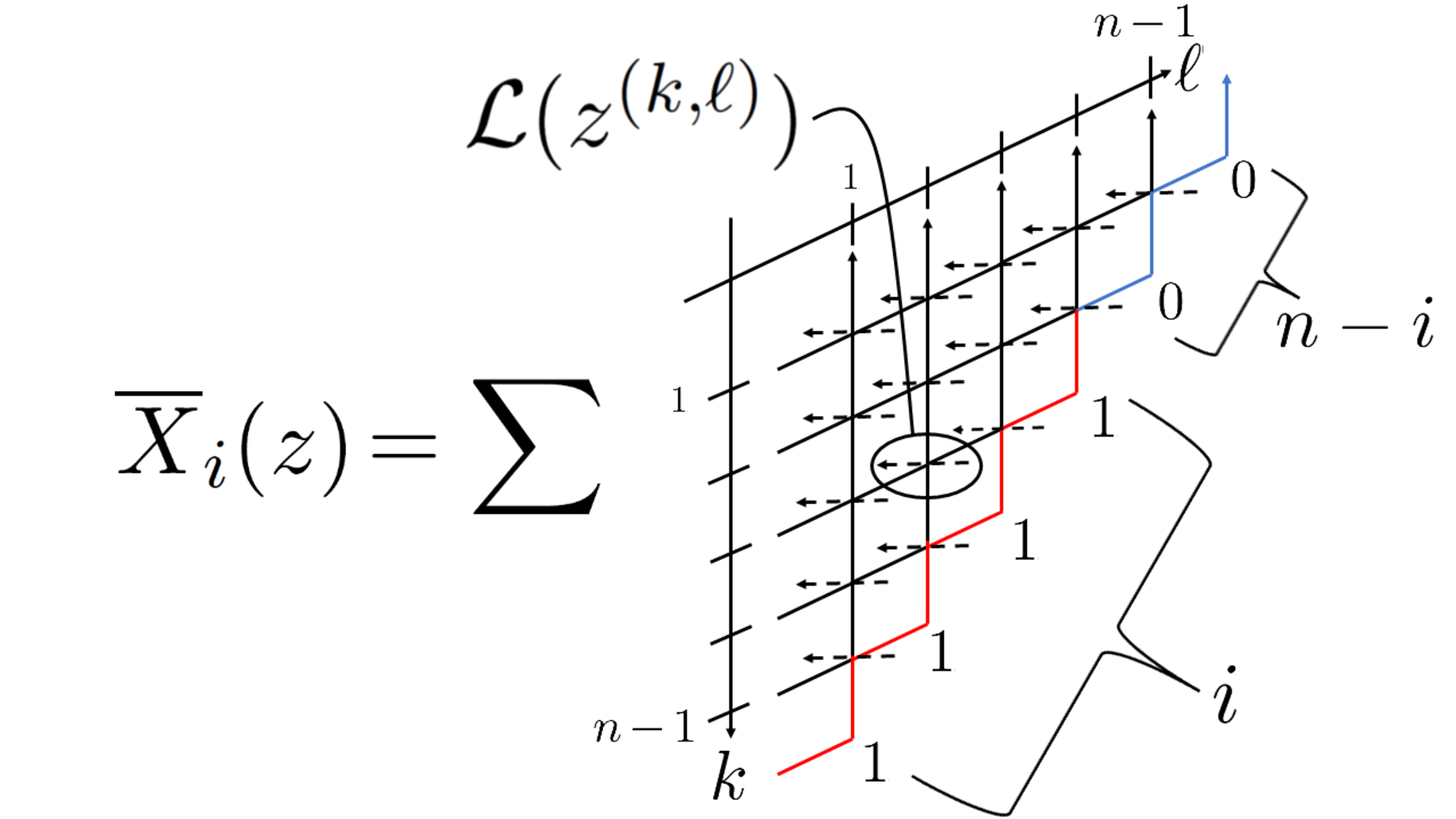}
\caption{The operator $\overline{X}_i(z)$.
We use $\mathcal{L}(z^{(k, \ell)})$ in the $k$-th row and $\ell$-th column $(1 \le k, \ell \le n-1, k+\ell \le n)$.
}
\label{barfigure}
\end{figure}

For a string ${\bf z}=(z_{1},z_{2}\dots,z_{m})$ of $m$ multiple variables, we define the operator
$\displaystyle \overline{X}_i({\bf z}):=
\overline{X}_i(z_{1})\overline{X}_i(z_{2})  \cdots \overline{X}_i(z_{m})$.

For $n$ strings of multiple variables $\mathbf{z}_i:=(z_{i1},z_{i2},\dots,z_{i,m_i})$, $(i=1,2,\dots,n,\ z_{ij}=(z_{ij}^{(k,\ell)})_{(k,\ell)\in D_n})$, we consider the partition function 
\begin{equation}\label{eq:partition_function_inhomogen}
\langle \Omega|   \overline{X}_n({\bf z}_1) \overline{X}_{n-1}({\bf z}_2) \cdots
\overline{X}_2({\bf z}_{n-1}) \overline{X}_0({\bf z}_n) |\Omega \rangle.
\end{equation}
Let $m_i=|\mathbf{z}_i|$ denote the number of elements of $\mathbf{z}_i$.
Note the operator \\
 $\overline{X}_n({\bf z}_1) \overline{X}_{n-1}({\bf z}_2) \cdots
\overline{X}_2({\bf z}_{n-1}) \overline{X}_0({\bf z}_n)$
depends on $\frac{(m_1+\dots+m_n)n(n-1)}{2}$ variables $z_{ij}^{(k,\ell)}$
in principle. However, as shown in the theorem below,
the partition functions
\eqref{eq:partition_function_inhomogen}
which is sandwiched by the vacuum state and its dual
actually depend on fewer variables.


\begin{theorem}
Assume ${\bf z}_i \neq \phi$ for all $i$. 
Let $M_0:=0$, $M_i:=m_1+m_2+\dots+m_i$,
$V:=M_{n-1}$, and $W:=M_n$.
Then, we have
\begin{align}
&\langle \Omega|   \overline{X}_n({\bf z}_1) \overline{X}_{n-1}({\bf z}_2) \cdots
\overline{X}_2({\bf z}_{n-1}) \overline{X}_0({\bf z}_n) |\Omega \rangle \nn \\
=&\frac{1}{\prod_{i=1}^{n-1} \prod_{j=1}^{m_i} z_{ij}^{(i,1)} }
\sum_{1 \le k_1 < k_2 < \cdots < k_V  \le W} \prod_{i=1}^{n-1} \prod_{j=1}^{m_i} 
w_{k_{M_{i-1}+j}}^{(i,1)},
\end{align}
where $(w_1,\dots,w_{W})$ is the concatenation of the strings: $(w_1,\dots,w_W)=\mathbf{z}_1\mathbf{z}_2\dots\mathbf{z}_{n}$.
\end{theorem}
In particular, when $m_1=m_2=\cdots=m_{n-1}=1$, we have $V=n-1$, $W=n-1+m_n$ and
\[
\sum_{1 \le k_1 < k_2 < \cdots < k_V  \le W} \prod_{i=1}^{n-1} \prod_{j=1}^{m_i} 
w_{k_{M_{i-1}+j}}^{(i,1)}
=\sum_{1 \le k_1 < k_2 < \cdots < k_{n-1}  \le n-1+m_n} \prod_{i=1}^{n-1} z_{k_i}^{(i,1)}.
\]
This polynomial is a special case of loop elementary symmetric functions 
\cite{Lam,LP,Yamada}.

Since the notations become complicated and can be proved in the same way
for the general case,
we show the following case corresponding to
$m_1=m_2=\cdots=m_{n-1}=1$.
\begin{theorem} \label{relationwithloop}
We have
\begin{align}
    &\langle \Omega|   \overline{X}_n(z_1) \overline{X}_{n-1}(z_2) \cdots
\overline{X}_2(z_{n-1}) \overline{X}_0({\bf z}_n) |\Omega \rangle 
=
\frac{1}{\prod_{i=1}^{n-1} z_{i}^{(i,1)} }
\sum_{1 \le k_1 < k_2 < \cdots < k_{n-1}  \le n-1+m_n} \prod_{i=1}^{n-1} z_{k_i}^{(i,1)}.\label{eq_to_show}
\end{align}
\end{theorem}

\begin{figure}[htbp]
\centering
\includegraphics[width=15truecm]{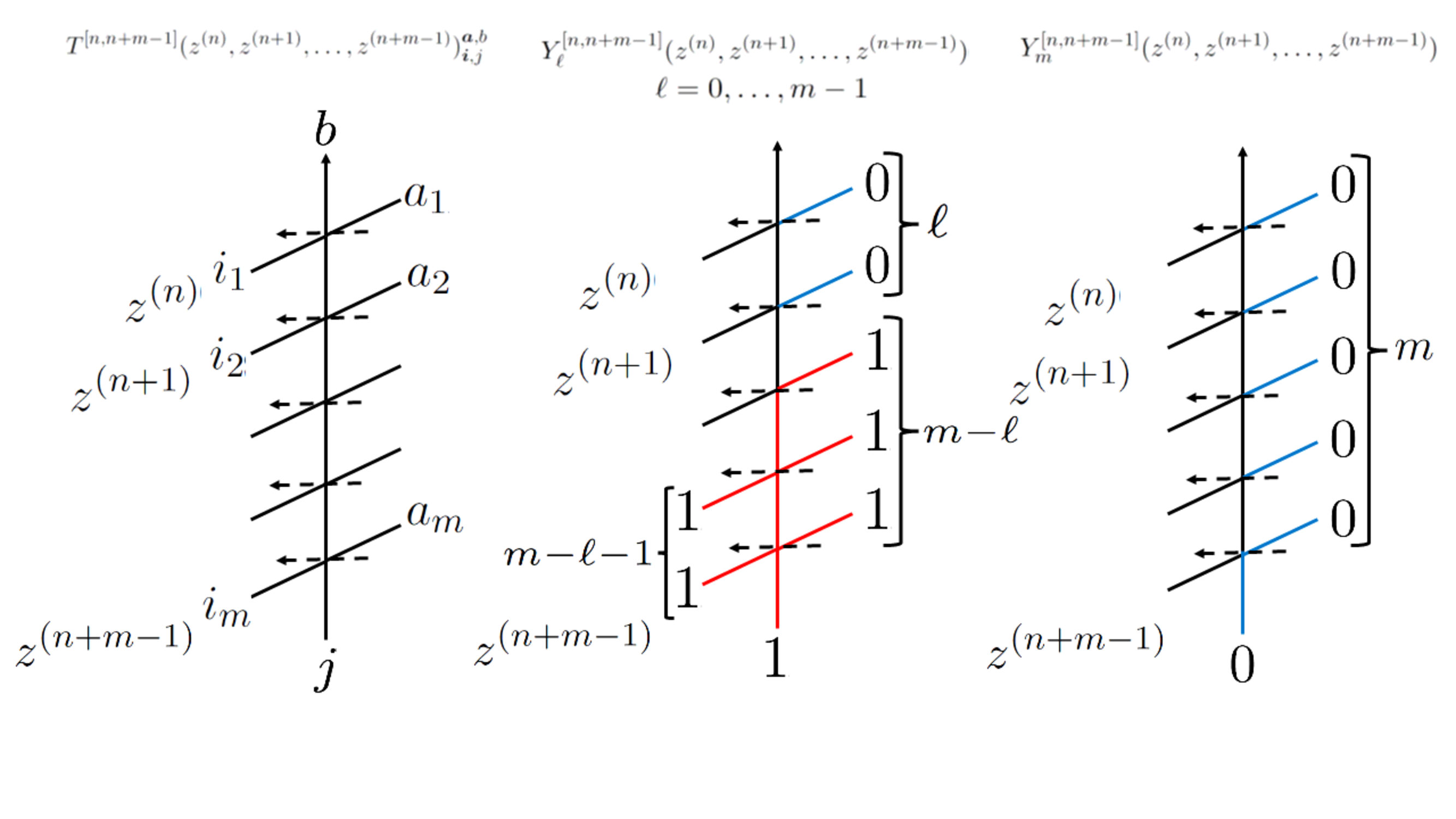}
\caption{
One-column operators $
T^{[n,n+m-1]}(z^{(n)},z^{(n+1)},\dots,z^{(n+m-1)})_{{\bs i},j}^{{\bs a},b}$
(left), $ Y_\ell^{[n,n+m-1]}(z^{(n)},z^{(n+1)},\dots,z^{(n+m-1)}) \ (
\ell=0,\dots,m-1)$ (middle) and $
Y_m^{[n,n+m-1]}(z^{(n)},z^{(n+1)},\dots,z^{(n+m-1)})$ (right).}
\label{onecolumnoperatorsfigure}
\end{figure}

Before proving Theorem \ref{relationwithloop}, let us introduce several notations.
Let $z=(z^{(k,l)})_{(k,l)\in D_n}$ be a multiple variable.
For brevity, we write $z^{(i)}=z^{(i,1)}$.
For two sets of indices ${\bs a}=(a_1,a_2,\dots,a_m)\in \{0,1\}^m$ and ${\bs i}=(i_1,i_2,\dots,i_m)\in \{0,1\}^m$, we introduce the linear operator
\begin{align}
&T^{[n,n+m-1]}(z^{(n)},z^{(n+1)},\dots,z^{(n+m-1)})_{{\bs i},j}^{{\bs a},b} \nonumber \\
:=&\sum_{k_1,\dots,k_{m-1}}
\mathcal{L}(z^{(n)})_{i_1 k_1}^{a_1 b}\otimes
\mathcal{L}(z^{(n+1)})_{i_2 k_2}^{a_2 k_1}\otimes
\cdots
\otimes
\mathcal{L}(z^{(n+m-2)})_{i_{m-1} k_{m-1}}^{a_{m-1} k_{m-2}}\otimes
\mathcal{L}(z^{(n+m-1)})_{i_m j}^{a_m k_{m-1}},
\end{align}
which acts on $(\mathcal{F}_{n1}\otimes \mathbb{C}[z^{(n)}])\otimes (\mathcal{F}_{n+1,1}\otimes \mathbb{C}[z^{(n+1)}])\otimes \dots \otimes (\mathcal{F}_{n+m-1,1}\otimes \mathbb{C}[z^{(n+m-1)}])$.

We also introduce the following operators
which are certain sums of the operators which we have just introduced
\begin{align}
&Y_\ell^{[n,n+m-1]}(z^{(n)},z^{(n+1)},\dots,z^{(n+m-1)}) \nn \\
&=\sum_{b,i_1,\dots,i_{\ell+1}} T^{[n,n+m-1]}(z^{(n)},z^{(n+1)},\dots,z^{(n+m-1)})_{(i_1,\dots,i_{\ell+1},1^{m-\ell-1}),1}^{(0^\ell,1^{m-\ell}),b},  
\end{align}
for $\ell=0,\dots,m-1$
and
\begin{align}
Y_m^{[n,n+m-1]}(z^{(n)},z^{(n+1)},\dots,z^{(n+m-1)})
=&\sum_{b,{\bs i}} T^{[n,n+m-1]}(z^{(n)},z^{(n+1)},\dots,z^{(n+m-1)})_{{\bs i},0}^{(0^{m}),b}.
\end{align}
Graphical descriptions of these operators are given
in Figure \ref{onecolumnoperatorsfigure}.

If there is no confusion, we use the abbreviations 
\[
\begin{gathered}
{T^{[n,n+m-1]}}(z)_{{\bs i},j}^{{\bs a},b}:=
T^{[n,n+m-1]}(z^{(n)},z^{(n+1)},\dots,z^{(n+m-1)})_{{\bs i},j}^{{\bs a},b},\\
Y_\ell^{[n,n+m-1]}(z):=
Y_\ell^{[n,n+m-1]}(z^{(n)},z^{(n+1)},\dots,z^{(n+m-1)}).
\end{gathered}
\]
Let
\[
{T^{(m)}}(z)_{{\bs i},j}^{{\bs a},b}:={T^{[1,m]}}(z)_{{\bs i},j}^{{\bs a},b},\qquad
Y_\ell^{(m)}(z):=Y_\ell^{[1,m]}(z).
\]



We also introduce the following notations for (dual) vectors
\begin{align}
| i_k,\dots,i_{k+m} \rangle_{[k,k+m]}&=|i_k \rangle_{k1} \otimes \cdots \otimes |i_{k+m} \rangle_{k+m,1}
\in \mathcal{F}_{k1} \otimes \cdots \otimes \mathcal{F}_{k+m,1}, \\
{}_{[k,k+m]} \langle i_k,\dots,i_{k+m} |
&={}_{k1} \langle i_k| \otimes \cdots \otimes {}_{k+m,1} \langle i_{k+m}|
\in \mathcal{F}^*_{k1} \otimes \cdots \otimes \mathcal{F}^*_{k+m,1}
.
\end{align}

\begin{figure}[htbp]
\centering
\includegraphics[width=12truecm]{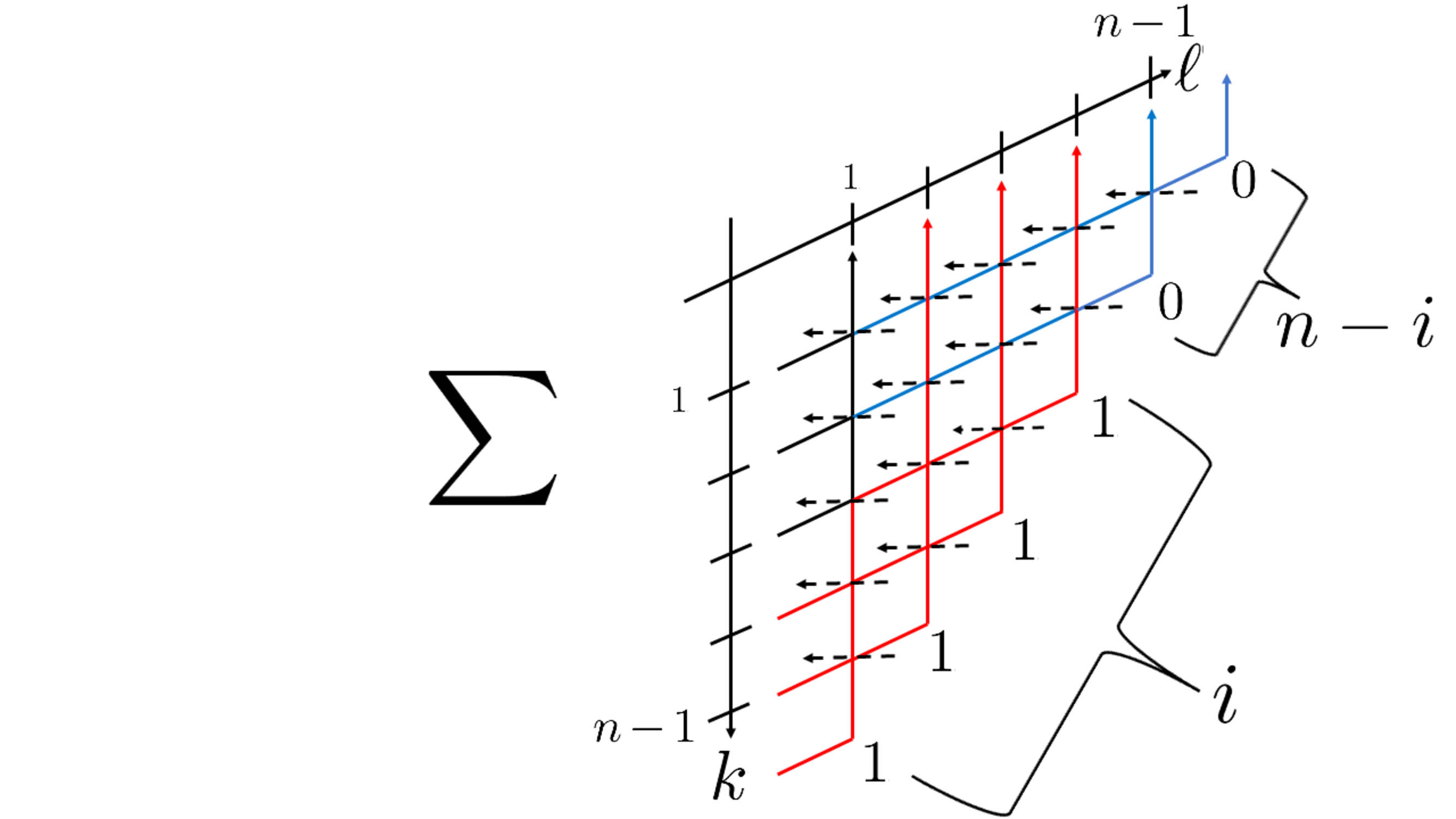}
\caption{The operator $\overline{X}_i(z_{n+1-i})$ $(2 \leq i \leq n)$ in
the inhomogeneous partition functions
$\langle \Omega|   \overline{X}_n(z_1) \overline{X}_{n-1}(z_2) \cdots
\overline{X}_2(z_{n-1}) 
\overline{X}_0(z_n) \cdots
\overline{X}_0(z_{n+k}) |\Omega \rangle$.
We find
the $L$-operators from the second to the $(n-1)$-th columns are all fixed,
and the $L$-operators in the first column form the operator
$Y_{n-i}^{(n-1)}(z_{n+1-i})$.
}
\label{barfreezefigure}
\end{figure}

\begin{lemma} \label{reductiontofirstcolumn}
For $k$ a non-negative integer, we have

\begin{align}
      &\langle \Omega|   \overline{X}_n(z_1) \overline{X}_{n-1}(z_2) \cdots
\overline{X}_2(z_{n-1}) 
\overline{X}_0(z_n) \cdots
\overline{X}_0(z_{n+k}) |\Omega \rangle \nn \\
=&{}_{[1,n-1]} \langle 0^{n-1}|Y_0^{(n-1)}(z_1) Y_1^{(n-1)}(z_2) \cdots  Y_{n-2}^{(n-1)}(z_{n-1})
Y_{n-1}^{(n-1)}(z_{n}) \cdots Y_{n-1}^{(n-1)}(z_{n+k})|0^{n-1} \rangle_{[1,n-1]}.
\label{reductionrelation}
\end{align}
\end{lemma}

\begin{proof}
This can be shown in 
a similar manner to the proof of
Lemma \ref{simplestpartition}.

Let $\mathfrak{Y}:=\overline{X}_n(z_1) \overline{X}_{n-1}(z_2) \cdots
\overline{X}_2(z_{n-1}) 
\overline{X}_0(z_n) \cdots
\overline{X}_0(z_{n+k})$.
By investigating the bosonic Fock spaces $\mathcal{F}_{k\ell}$,
$(k,\ell)=(n-2,2) \to (n-3,3) \to (n-3,2) \to (n-4,4) \to (n-4,3) \to (n-4,2) \to \cdots \to (1,2)$,
we find that configurations that survive in $\langle \Omega|\mathfrak{Y}|\Omega\rangle$ should satisfy the following property: 
On the slice that corresponds to $\overline{X}_i(z_{n+1-i})$ $(2 \leq i \leq n)$, the matrix element assigned to $(k,l)$ is (i) $1_r$ if $k>n-i$, and (ii) $\mathbf{t}$ if $k\leq n-i$ and $l>1$.
(Figure \ref{barfreezefigure} presents a such configuration.)
This fact implies that $\langle \Omega|\mathfrak{Y}|\Omega\rangle$ can
be factorized as a product
of the partition function constructed only from the first column
and that constructed from the other columns.
Since the latter becomes $\prod_{(k,l)\in D_n,k\neq 1} {}_{k\ell}\langle 0|0\rangle_{k\ell}=1$, $\langle \Omega|\mathfrak{Y}|\Omega\rangle$ is equal to the one constructed from the first column,
i.e. we note that the three-dimensional partition functions are reduced to
partition functions
\begin{align}
    {}_{[1,n-1]} \langle 0^{n-1}|Y_0^{(n-1)}(z_1) Y_1^{(n-1)}(z_2) \cdots  Y_{n-2}^{(n-1)}(z_{n-1})
Y_{n-1}^{(n-1)}(z_{n}) \cdots Y_{n-1}^{(n-1)}(z_{n+k})|0^{n-1} \rangle_{[1,n-1]},\label{partition_function_a}
\end{align}
consisting only of $L$-operators in
the first column, i.e., operators $Y_{n-i}^{(n-1)}(z)$,
sandwiched by reduced (dual) vacuum vectors ${}_{[1,n-1]} \langle 0^{n-1}|$,
$|0^{n-1} \rangle_{[1,n-1]}$ which are reductions of the original vacuum vectors $\langle \Omega|, |\Omega \rangle$.
\end{proof}


\begin{figure}[htbp]
\centering
\includegraphics[width=15truecm]{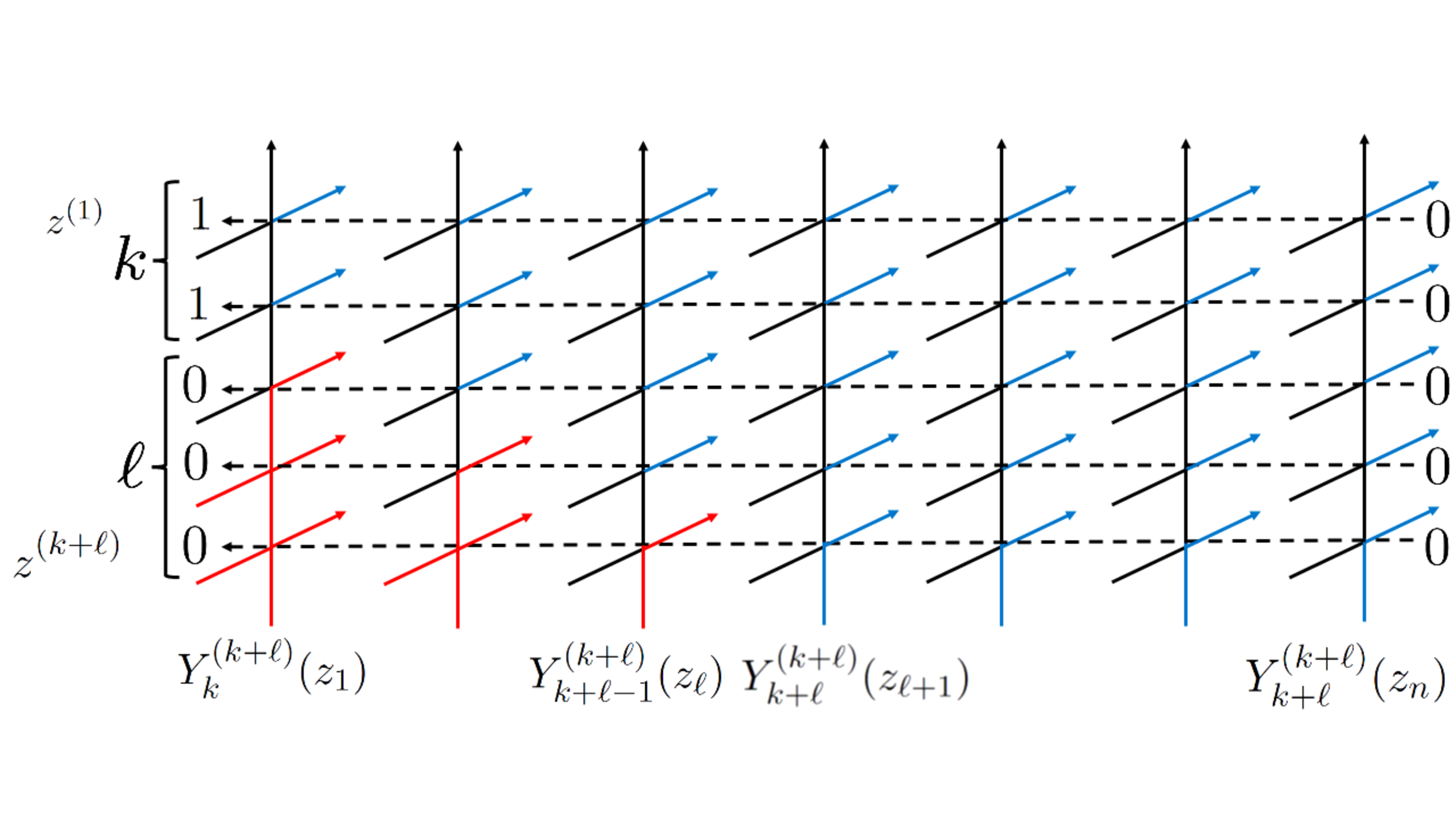}
\caption{
Another class of partition functions
constructed from one-row operators
${}_{[1,k+\ell]} \langle 1^k,0^\ell|
Y_k^{(k+\ell)}(z_1) Y_{k+1}^{(k+\ell)}(z_2) \cdots
Y_{k+\ell-1}^{(k+\ell)}(z_\ell) Y_{k+\ell}^{(k+\ell)}(z_{\ell+1})
\cdots Y_{k+\ell}^{(k+\ell)}(z_n)  |0^{k+\ell} \rangle_{[1,k+\ell]}$.
}
\label{onecolumnpartitionfunctionsbasicfigure}
\end{figure}

From Lemma \ref{reductiontofirstcolumn},
the proof of Theorem \ref{relationwithloop}
reduces to showing that
the right-hand side of \eqref{reductionrelation}
is given by the loop elementary symmetric function \eqref{eq_to_show}. 
For this, we prepare the following Lemma \ref{fundlemma} about
some special type of
partition functions constructed from one-column operators
(Figure \ref{onecolumnpartitionfunctionsbasicfigure})
and loop elementary symmetric functions.

\begin{figure}[htbp]
\centering
\includegraphics[width=10truecm]{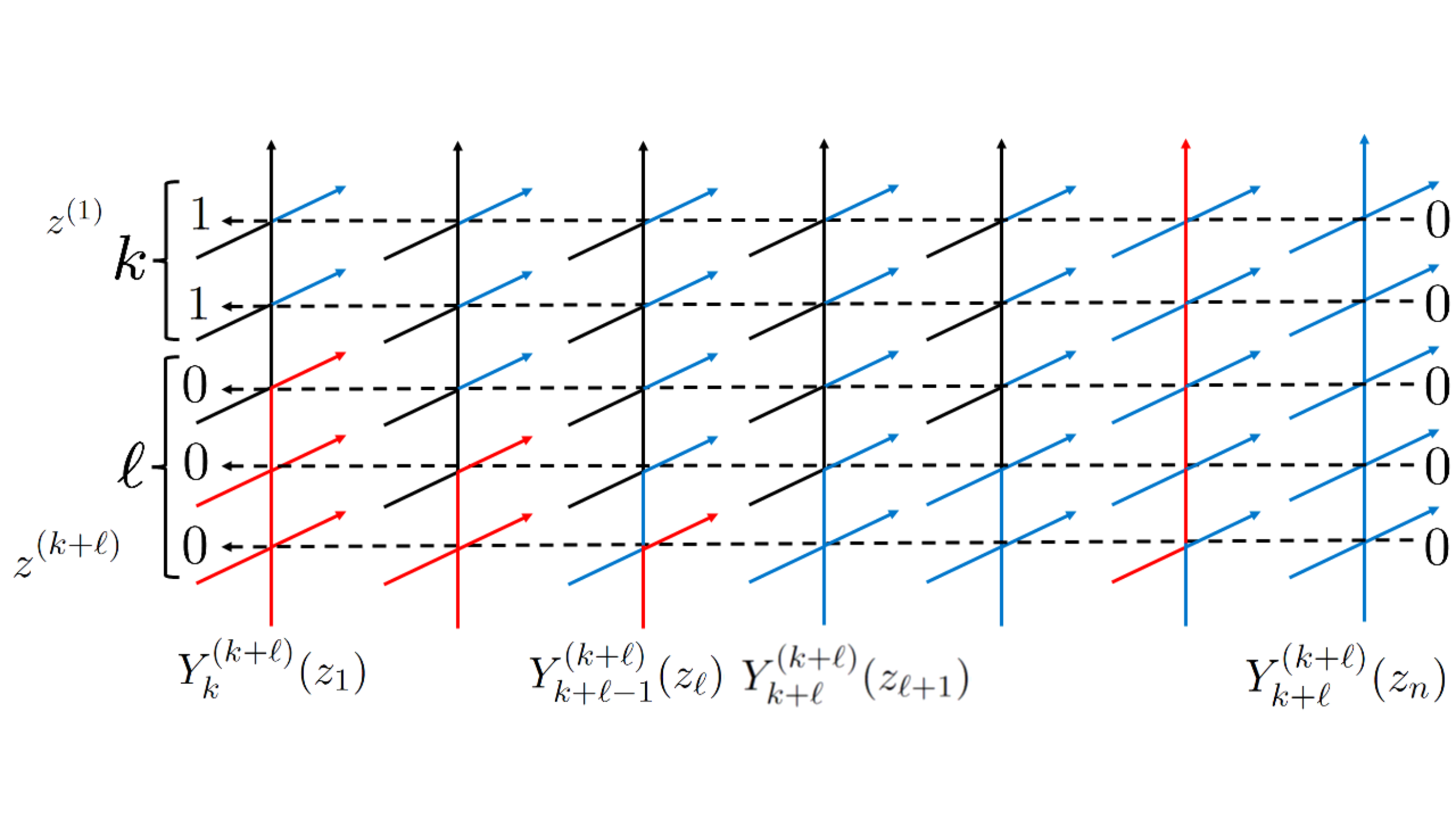}
\includegraphics[width=10truecm]{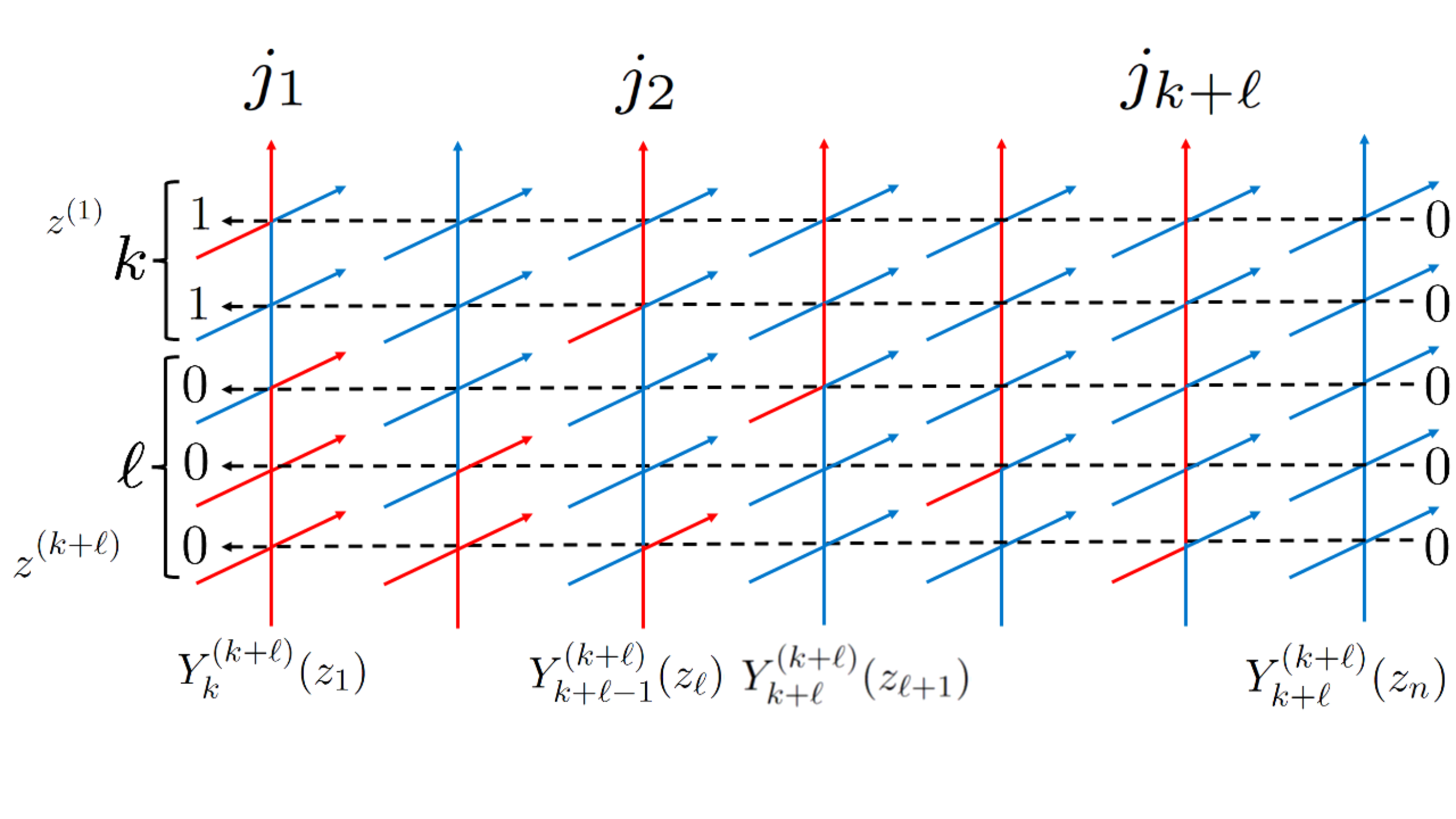}
\caption{Configurations which contribute to the partition functions.
Investigating the edges connecting the last two dashed lines
(with which $z^{(k+\ell-1)}$ and $z^{(k+\ell)}$ are associated) 
which were not colored in Figure \ref{onecolumnpartitionfunctionsbasicfigure},
we find only one of the edges is colored with red,
and the others are colored with blue.
We also note that starting from that red edge,
there is a straight red line passing through that column,
and the columns right to that column are all colored with blue (top panel).
Repeating this argument, we find every nontrivial configuration
can be graphically represented as the bottom panel.
}
\label{onecolumnpartitionfunctionsbasicforprooffigure}
\end{figure}

\begin{lemma} \label{fundlemma}
For $k \geq 1$ and $k+\ell \geq n$, we have
\begin{align}
&{}_{[1,k+\ell]} \langle 1^k,0^\ell|
Y_k^{(k+\ell)}(z_1) Y_{k+1}^{(k+\ell)}(z_2) \cdots
Y_{k+\ell-1}^{(k+\ell)}(z_\ell) Y_{k+\ell}^{(k+\ell)}(z_{\ell+1})
\cdots Y_{k+\ell}^{(k+\ell)}(z_n)  |0^{k+\ell} \rangle_{[1,k+\ell]} \nn \\
=&(z_1^{(k+1)} \cdots z_\ell^{(k+\ell)})^{-1} \sum_{1 \le j_1 < j_2 < \cdots < j_{k+\ell} \le n } z_{j_1}^{(1)} \cdots z_{j_{k+\ell}}^{(k+\ell)}.
\label{auxrel}
\end{align}
\end{lemma}

\begin{proof}

We present a proof based on a graphical representation.
See
Figure \ref{onecolumnpartitionfunctionsbasicfigure}.
First, we observe that from the $\ell$-th to the $n$-th columns,
only one of the edges with which $z^{(k+\ell-1)}$ and $z^{(k+\ell)}$ are associated
connecting the last two dashed lines
is colored with red, and the rest are colored with blue.
One can observe this as follows.
In the bottom layer,
one can see from Figure \ref{onecolumnpartitionfunctionsbasicfigure}
that
there can be at most one annihilation
operator
coming from the $\ell$-th column.
If there is no annihilation operator,
then the edges in the $\ell$-th column are colored with red,
and the edges in the $(\ell+1) \sim n$-th columns are colored with blue since the operators are sandwiched by 
$|0 \rangle$ and $\langle 0|$ and
neither creation nor annihilation operators can appear.
When there is one annihilation operator arising
from the $\ell$-th column,
there must be exactly one annihilation operator
arsing from one of the $(\ell+1) \sim n$-th columns
which the uncolored edges must be colored with red,
and the other columns corresponding to identity operators
must be colored with blue.

Next, we note that colors of edges in the column
layer which contains the red-colored connecting edge
are uniquely determined, and then note all colors
in the column layers right to that layer
are determined to be blue (Figure \ref{onecolumnpartitionfunctionsbasicforprooffigure}, top).
Let us see this in some more detail.
When looking the $(k+\ell-1)$-th row of the column
layer which contains the red-colored connecting edge,
we note the operator-valued $L$-operator assigned
has the form $\mathcal{L}_{i1}^{0j}$.
In this case, for this to be nonzero, $i$ and $j$
must be $i=0$, $j=1$. One can continue this observation,
and we find there is a red line passing through the column,
and zero-number projection operators arise from the
first to the $(k+\ell-1)$-th row.
From this, we note the edges in the column layers right to that layer are all colored with blue.
This can be understood in the same way.
When looking the $(k+\ell-1)$-th row of those column
layers,
we note every operator-valued $L$-operator assigned
has the form $\mathcal{L}_{i0}^{0j}$.
For this to be nonzero,
either $i=j=0$ which corresponds to the identity operator
or $i=j=1$ which corresponds to the creation operator.
However, since we have to multiply by the zero-number projection operators after all, using
$i=j=1$ case does not contribute
to the partition function, and one gets $i=j=0$.
One can repeat this argument
and find all edges in those layers
are colored with blue.

Repeating the argument above, we find that all admissible
configurations can be represented as Figure \ref{onecolumnpartitionfunctionsbasicforprooffigure}, bottom.
There are $k+\ell$ red lines
partially passing through the column, which
can be labelled by $k+\ell$ integers $1 \le j_1<j_2<\dots<j_{k+\ell} \le n$.
Let us denote the contribution to the partition function
coming from this configuration as $Z_{j_1,j_2,\dots,j_{k+\ell}}$.

Computing the products of weights for this configuration,
we find
\begin{align}
Z_{j_1,j_2,\dots,j_{k+\ell}}=(z_1^{(k+1)} \cdots z_\ell^{(k+\ell)})^{-1}  z_{j_1}^{(1)} \cdots z_{j_{k+\ell}}^{(k+\ell)}.
\end{align}
Taking all configurations into account,
i.e. taking sum
$\sum_{1 \le j_1<j_2<\dots<j_{k+\ell} \le n}
Z_{j_1,j_2,\dots,j_{k+\ell}}
$ gives the explicit expression of the partition function
\eqref{auxrel}.
\end{proof}

Finally, we show the partition functions
given in Figure \ref{onecolumnpartitionfunctionsfigure}
give loop elementary functions.

\begin{figure}[htbp]
\centering
\includegraphics[width=15truecm]{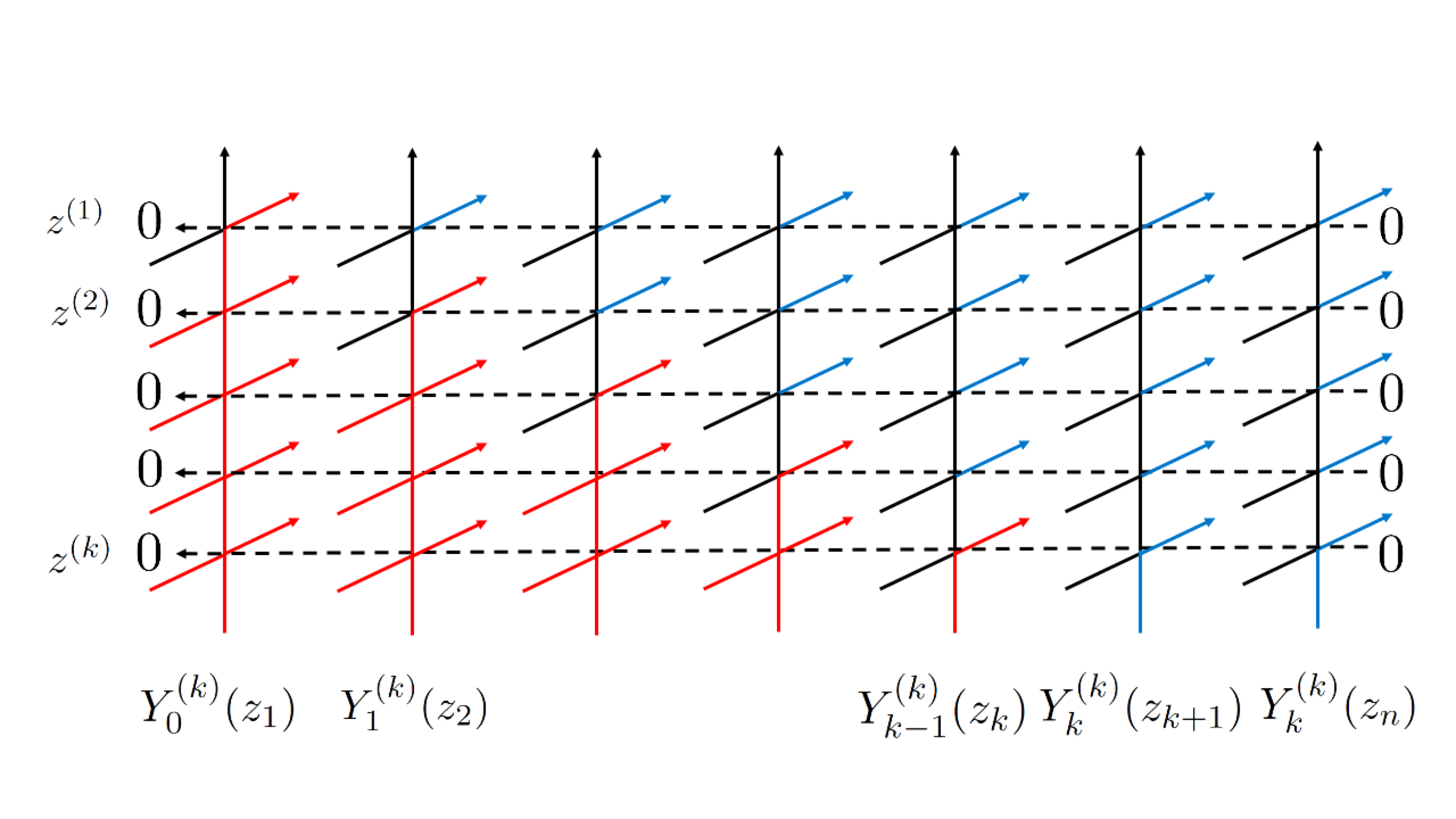}
\caption{
Partition functions ${}_{[1,k]} \langle 0^k|Y_0^{(k)}(z_1) Y_1^{(k)}(z_2) \cdots  Y_{k-1}^{(k)}(z_{k})
Y_k^{(k)}(z_{k+1}) \cdots Y_k^{(k)}(z_{n})|0^k \rangle_{[1,k]}$
constructed from one-column operators.}
\label{onecolumnpartitionfunctionsfigure}
\end{figure}

\begin{proposition} \label{onecolumnpartitionfunctionsloop}
We have
\begin{align}
&{}_{[1,k]} \langle 0^k|Y_0^{(k)}(z_1) Y_1^{(k)}(z_2) \cdots  Y_{k-1}^{(k)}(z_{k})
Y_k^{(k)}(z_{k+1}) \cdots Y_k^{(k)}(z_{n})|0^k \rangle_{[1,k]} \nn \\
=&(z_1^{(1)} \cdots z_k^{(k)})^{-1} \sum_{1 \le j_1 < j_2 < \cdots < j_k \le n } z_{j_1}^{(1)} \cdots z_{j_k}^{(k)}.
\label{fundrel}
\end{align}
\end{proposition}

\begin{proof}
We prove by induction on $k$
with the help of Lemma \ref{fundlemma}.
Suppose \eqref{fundrel} with $k$ replaced by $k-1$ holds.
For $m\geq 0$, consider the projection operator
\[
|m\rangle_1{}_1\langle m|:
\mathcal{F}_{11}\otimes \mathcal{F}_{21}\otimes \dots\otimes \mathcal{F}_{k1}
\to 
\mathcal{F}_{11}\otimes \mathcal{F}_{21}\otimes \dots\otimes \mathcal{F}_{k1},
\]
which sends $|i_1,i_2,\dots,i_{k}\rangle_{[1,k]}$ to 
$|m,i_2,\dots,i_{k}\rangle_{[1,k]}\cdot \langle m|i_1\rangle $.
Then the identity operator $\mathrm{id}$ decomposes into 
\begin{align}
\mathrm{id}=\sum_{m\geq 0}|m\rangle_1{}_1\langle m|. \label{identityop}
\end{align}
By using \eqref{identityop}, we decompose the left-hand side of \eqref{fundrel}
as
\begin{align}
&{}_{[1,k]} \langle 0^k|Y_0^{(k)}(z_1) | 0 \rangle_1
{}_1 \langle 0 | Y_1^{(k)}(z_2) \cdots  Y_{k-1}^{(k)}(z_{k})
Y_k^{(k)}(z_{k+1}) \cdots Y_k^{(k)}(z_{n})|0^k \rangle_{[1,k]} \nn \\
+&\sum_{i=2}^{k-1} {}_{[1,k]} \langle 0^k|Y_0^{(k)}(z_1)
| 1 \rangle_1 {}_1 \langle 1 |
 Y_1^{(k)}(z_2) | 1 \rangle_1 \cdots  
{}_1 \langle 1 |
 Y_{i-2}^{(k)}(z_{i-1}) | 1 \rangle_1
{}_1 \langle 1 |
 Y_{i-1}^{(k)}(z_{i}) | 0 \rangle_1 \nn \\
&\times
{}_1 \langle 0| Y_i^{(k)}(z_{i+1}) \cdots Y_{k-1}^{(k)}(z_{k})
Y_k^{(k)}(z_{k+1}) \cdots Y_k^{(k)}(z_{n})|0^k \rangle_{[1,k]} \nn \\
+&\sum_{i=k}^{n} {}_{[1,k]} \langle 0^k|Y_0^{(k)}(z_1)
| 1 \rangle_1 {}_1 \langle 1 |
 Y_1^{(k)}(z_2) | 1 \rangle_1 \cdots  
{}_1 \langle 1 |
 Y_{k-1}^{(k)}(z_{k}) | 1 \rangle_1 \nn \\
&\times {}_1 \langle 1 |
 Y_{k}^{(k)}(z_{k+1}) | 1 \rangle_1 \cdots
{}_1 \langle 1 |
 Y_{k}^{(k)}(z_{i}) | 0 \rangle_1
{}_1 \langle 0| Y_k^{(k)}(z_{i+1}) \cdots  Y_k^{(k)}(z_{n})|0^k \rangle_{[1,k]}.
\label{decompfund}
\end{align}
Note that here we use the fact that one can further restrict the summands of the identity
operator \eqref{identityop} only to $|0 \rangle \langle 0|$
and $|1 \rangle \langle 1|$.
This is because in the top slice,
only the first column can produce an annihilation operator
(see Figure \ref{onecolumnpartitionfunctionsfigure}), and all the operators
coming from the first row
are finally sandwiched
by $|0 \rangle$ and $\langle 0|$,
which means $|m \rangle$ $(m \geq 2)$ do not appear
since we cannot change to $|0 \rangle$ by acting
just one annihilation operator, 
and we can neglect $|m \rangle \langle m|$ $(m \geq 2)$.

We first examine the first term in \eqref{decompfund}.
One can see the following relations hold
\begin{align}
{}_{[1,k]} \langle 0^k|Y_0^{(k)}(z_1) | 0 \rangle_1={}_{[2,k]} \langle 0^{k-1}|,
\label{fundfirsttermfirstfactor}
\end{align}
and
\begin{align}
&{}_{[2,k]} \langle 0^{k-1}| {}_1 \langle 0 | Y_1^{(k)}(z_2) \cdots  Y_{k-1}^{(k)}(z_{k})
Y_k^{(k)}(z_{k+1}) \cdots Y_k^{(k)}(z_{n})|0^k \rangle_{[1,k]} \nn \\
&=
{}_{[2,k]} \langle 0^{k-1}| Y_1^{(k-1)}(z_2') \cdots  Y_{k-1}^{(k-1)}(z_{k}')
Y_k^{(k-1)}(z_{k+1}') \cdots
Y_k^{(k-1)}(z_{n}')|0^{k-1} \rangle_{[2,k]}.
\label{fundfirsttermsecondfactor}
\end{align}
Here, $z_j'=( z_j^{(2)},z_j^{(3)},\dots,z_j^{(k)})$.


One can apply the inductive assumption to
the right hand side of \eqref{fundfirsttermsecondfactor}, which is
\begin{align}
&{}_{[2,k]} \langle 0^{k-1}| Y_1^{(k-1)}(z_2') \cdots  Y_{k-1}^{(k-1)}( z_{k}' )
Y_k^{(k-1)}( z_{k+1}' )
\cdots Y_k^{(k-1)}(z_{n}'  )|0^{k-1} \rangle_{[2,k]} \nn \\
&=(z_2^{(2)} \cdots z_k^{(k)})^{-1} \sum_{2 \le  j_2 < \cdots < j_k \le n } z_{j_2}^{(2)} \cdots z_{j_k}^{(k)},
\label{firsttermexplicit}
\end{align}
hence combining with
\eqref{fundfirsttermfirstfactor} and
\eqref{fundfirsttermsecondfactor}, we find the first term of the left hand side of \eqref{decompfund}
is \eqref{firsttermexplicit}.

\begin{figure}[htbp]
\centering
\includegraphics[width=10truecm]{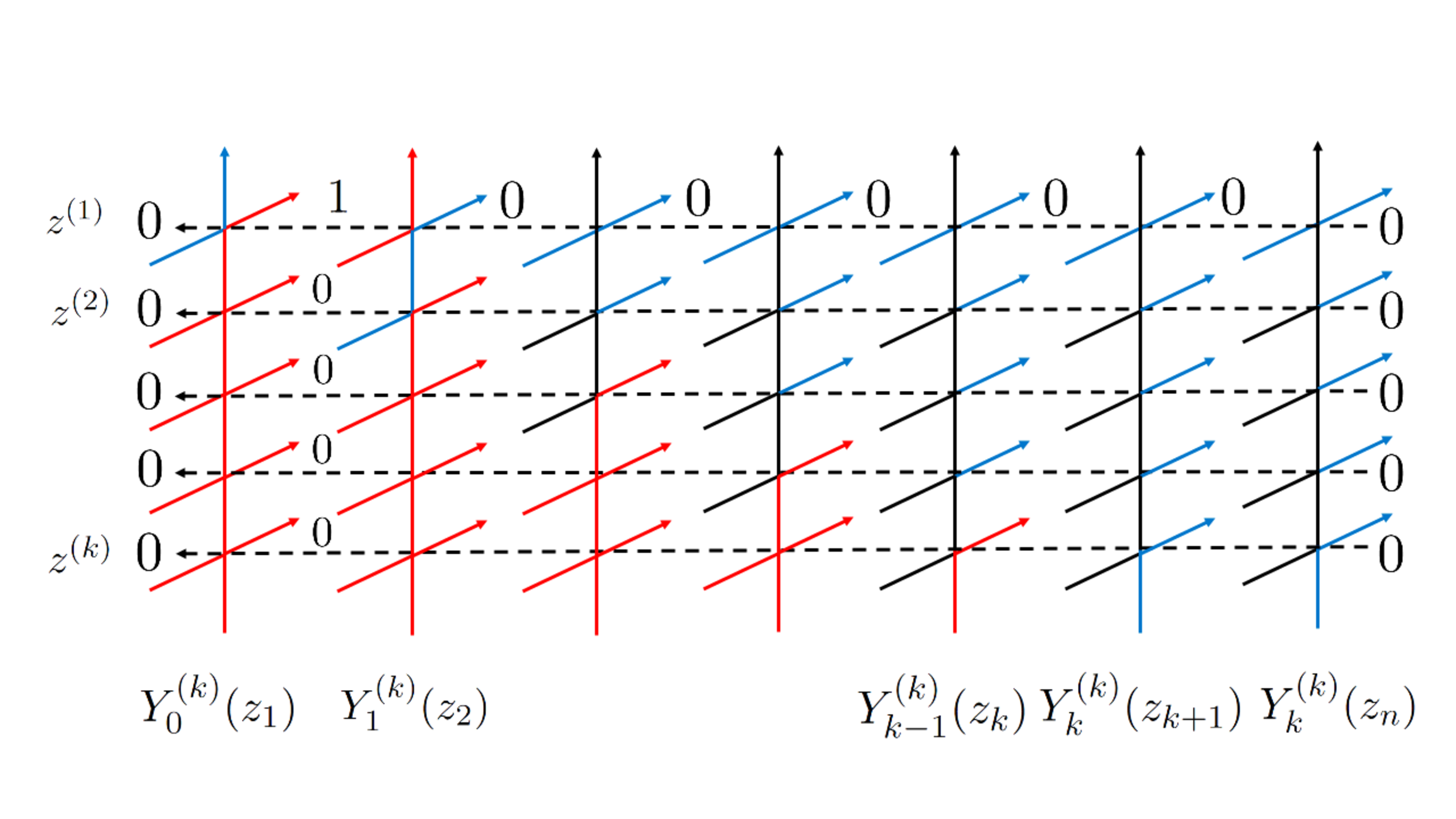}
\includegraphics[width=10truecm]{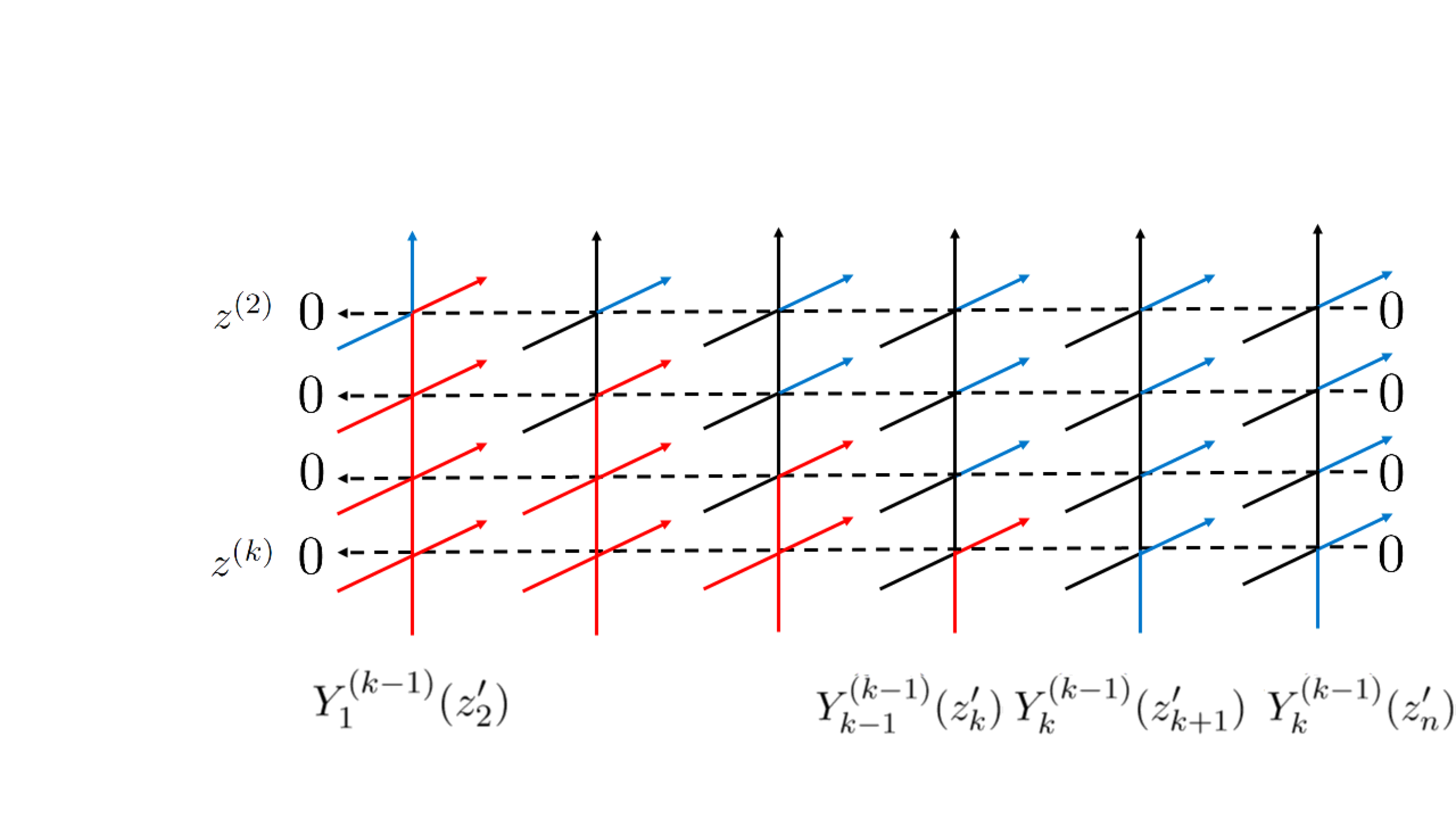}
\caption{  The first term in \eqref{decompfund}
is the case when 1 appears in the leftmost inner part
in the top dashed line.
After coloring parts in the solid lines which are already fixed to be red or blue,
we get the top panel.
We can remove the leftmost column layer as everything is fixed.
We can also remove the top row layer as the matrix elements of the $L$-operators from the third to the last columns are all identity operators.
Removing the leftmost column and the top row layers,
we get the bottom panel, which is nothing but \eqref{firsttermexplicit}.
}
\label{onecolumnpartitionfunctionsdivonefigure}
\end{figure}

Figure \ref{onecolumnpartitionfunctionsdivonefigure}
gives a graphical understanding of
the computations above to get \eqref{firsttermexplicit}.

Next we deal with the summands in \eqref{decompfund}.
Let us see the summands in the first sum.
One can show the first factor is
\begin{align}
{}_{[1,k]} \langle 0^k|Y_0^{(k)}(z_1)
| 1 \rangle_1=(z_1^{(1)})^{-1} {}_{[2,k]} \langle 0^{k-1}|.
\end{align}
Acting the second factor on ${}_{[2,k]} \langle 0^{k-1}|$,
we get
\begin{align}
{}_{[2,k]} \langle 0^{k-1}| {}_1 \langle 1 |
 Y_1^{(k)}(z_2) | 1 \rangle_1
=(z_2^{(2)})^{-1} {}_{[2,k]} \langle 1,0^{k-2}|.
\end{align}
One can repeat this computation. Using
\begin{align}
{}_{[2,k]} \langle 1^j, 0^{k-1-j}| {}_1 \langle 1 |
 Y_1^{(k)}(z_{2+j}) | 1 \rangle_1
&=(z_{2+j}^{(2+j)})^{-1} {}_{[2,k]} \langle 1^{1+j},0^{k-2-j}|, \ \ \ j=0,\dots,i-3, \\
{}_{[2,k]} \langle 1^{i-2}, 0^{k+1-i}| {}_1 \langle 1 |
 Y_1^{(k)}(z_{i}) | 1 \rangle_1&=
z_i^{(1)} (z_i^{(i)})^{-1} {}_{[2,k]} \langle 1^{i-1}, 0^{k-i}|,
\end{align}
we get
\begin{align}
&{}_{[1,k]} \langle 0^k|Y_0^{(k)}(z_1)
| 1 \rangle_1 {}_1 \langle 1 |
 Y_1^{(k)}(z_2) | 1 \rangle_1 \cdots  
{}_1 \langle 1 |
 Y_{i-2}^{(k)}(z_{i-1}) | 1 \rangle_1
{}_1 \langle 1 |
 Y_{i-1}^{(k)}(z_{i}) | 0 \rangle_1 \nn \\
=&z_i^{(1)} (z_1^{(1)} \cdots z_i^{(i)})^{-1} {}_{[2,k]} \langle 1^{i-1}, 0^{k-i}|.
\label{onefactorfinal}
\end{align}
Next we note that
\begin{align}
&{}_{[2,k]}
\langle 1^{i-1}, 0^{k-i}|
{}_1 \langle 0| Y_i^{(k)}(z_{i+1}) \cdots Y_{k-1}^{(k)}(z_{k})
Y_k^{(k)}(z_{k+1}) \cdots Y_k^{(k)}(z_{n})|0^k \rangle_{[1,k]} \nn \\
=&{}_{[2,k]}
\langle 1^{i-1}, 0^{k-i}|
Y_i^{(k-1)}(z_{i+1}'  ) \cdots Y_{k-1}^{(k-1)}(z_{k}'  ) Y_k^{(k-1)}(z_{k+1}'   ) \cdots Y_k^{(k-1)}(z_{n}'      )|0^{k-1} \rangle_{[2,k]}.
\label{anotherfactor}
\end{align}
Applying \eqref{auxrel} in Lemma \ref{fundlemma} to the right hand side of
\eqref{anotherfactor}, we get
\begin{align}
&{}_{[2,k]}
\langle 1^{i-1}, 0^{k-i}|
Y_i^{(k-1)}(z_{i+1}') \cdots Y_{k-1}^{(k-1)}(z_{k}'  ) Y_k^{(k-1)}(z_{k+1}'  ) \cdots Y_k^{(k-1)}(z_{n}'     )|0^{k-1} \rangle_{[2,k]} 
\nn \\
=&(z_{i+1}^{(i+1)} \cdots z_k^{(k)})^{-1} \sum_{i+1 \le j_2 < \cdots < j_k \le n} z_{j_2}^{(2)} \cdots z_{j_k}^{(k)}.
\label{secondtermafterremoval}
\end{align}
Together with \eqref{onefactorfinal} and  \eqref{anotherfactor},
the first sum of the left hand side of \eqref{decompfund} is given by
\begin{align}
\sum_{i=2}^{k-1}
(z_{1}^{(1)} \cdots z_k^{(k)})^{-1} z_{i}^{(1)} \sum_{i+1 \le j_2 < \cdots < j_k \le n} z_{j_2}^{(2)} \cdots z_{j_k}^{(k)}.
\label{secondtermexplicit}
\end{align}

\begin{figure}[htbp]
\centering
\includegraphics[width=10truecm]{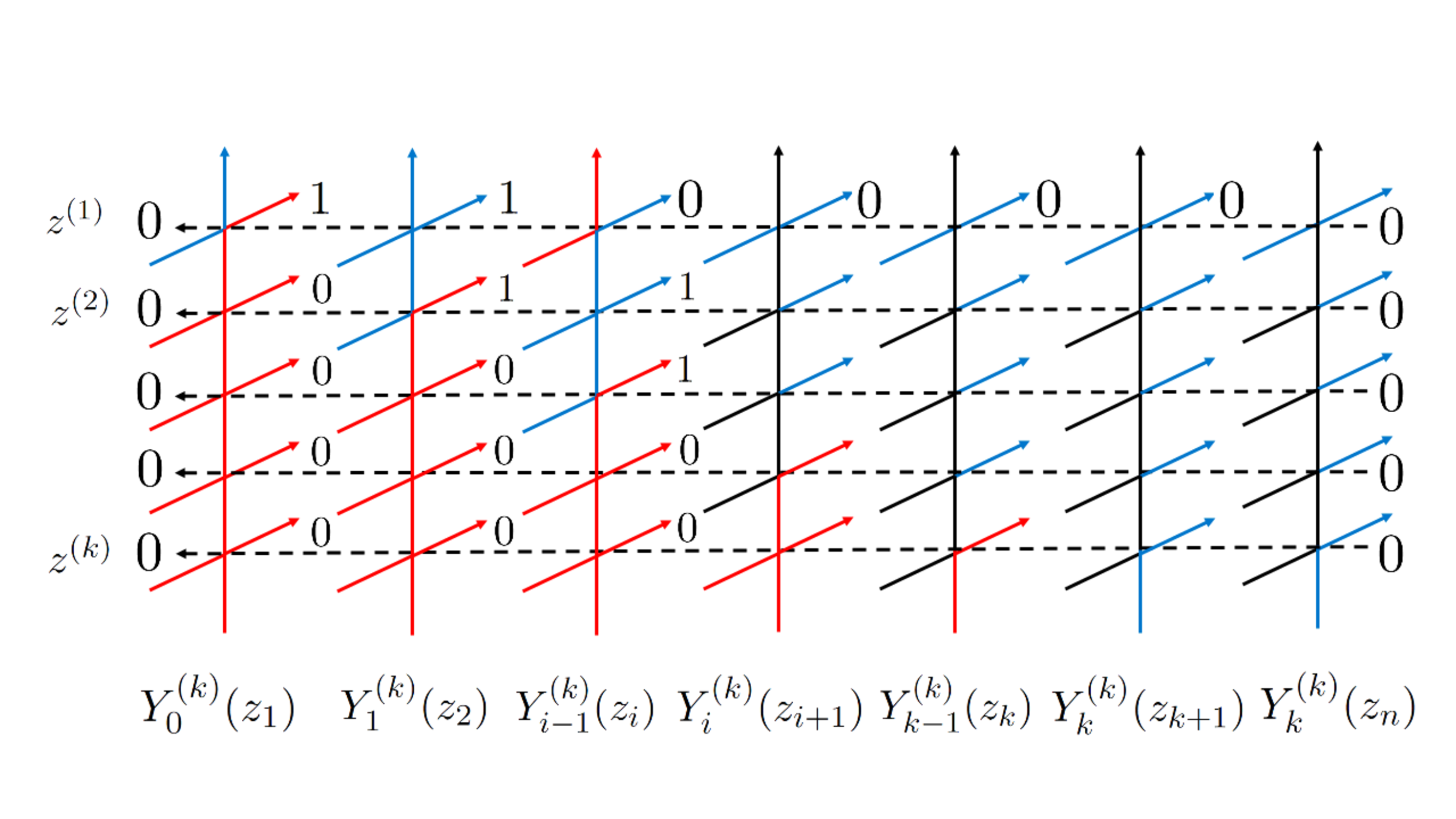}
\includegraphics[width=10truecm]{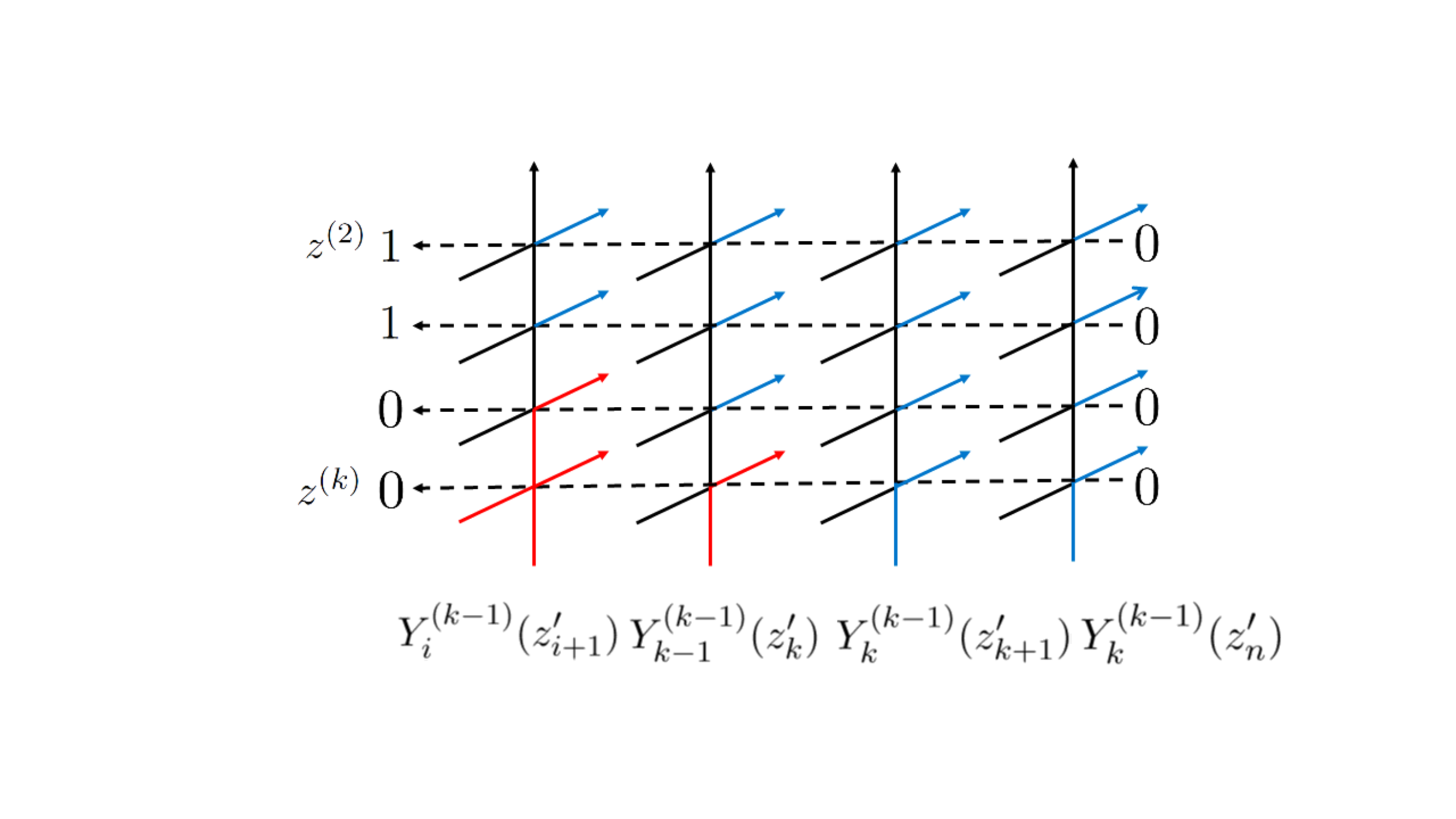}
\caption{  A summand of the first sum in \eqref{decompfund}
is graphically represented as the top panel
after coloring parts in the solid lines which are already fixed to be red or blue.
Removing the leftmost column and the top row layers,
we get the bottom panel,
which corresponds to \eqref{secondtermafterremoval}.
The removed parts contribute the factor
$z_i^{(1)} (z_1^{(1)} \cdots z_i^{(i)})^{-1}$.
Multiplying with \eqref{secondtermafterremoval} gives
\eqref{secondtermexplicit}.
}
\label{onecolumnpartitionfunctionsdivtwofigure}
\end{figure}
See
Figure \ref{onecolumnpartitionfunctionsdivtwofigure}
for a graphical derivation of \eqref{secondtermexplicit}.

In the same way, we can prove that 
the second sum of the left hand side of \eqref{decompfund} is given by
\begin{align}
\sum_{i=k}^{n-k+1}
(z_{1}^{(1)} \cdots z_k^{(k)})^{-1} z_{i}^{(1)} \sum_{i+1 \le j_2 < \cdots < j_k \le n} z_{j_2}^{(2)} \cdots z_{j_k}^{(k)}.
\label{thirdtermexplicit}
\end{align}
Finally, we note that the sum of
\eqref{firsttermexplicit},
\eqref{secondtermexplicit}  and \eqref{thirdtermexplicit} can be written as
\begin{align}
&(z_1^{(1)} \cdots z_k^{(k)})^{-1} \sum_{1 \le j_1 < j_2 < \cdots < j_k \le n } z_{j_1}^{(1)} \cdots z_{j_k}^{(k)},
\end{align}
which completes the proof of \eqref{fundrel}.

\end{proof}

Combining Lemma \ref{reductiontofirstcolumn}
and  Proposition
\ref{onecolumnpartitionfunctionsloop},
we get Theorem \ref{relationwithloop}.

\section*{Acknowledgement}
S.I.
is supported by Grant-in-Aid for Scientific Research
19K03605, 22K03239, 23K03056, JSPS. 
K.M.
is supported by Grant-in-Aid for Scientific Research 21K03176, 20K03793, JSPS.
R.O. is
supported by Grant-in-Aid for Scientific Research 21K03180, JSPS.

\section*{Appendix A: Tetrahedron equation}

\begin{figure}[htbp]
\centering
\includegraphics[width=12truecm]{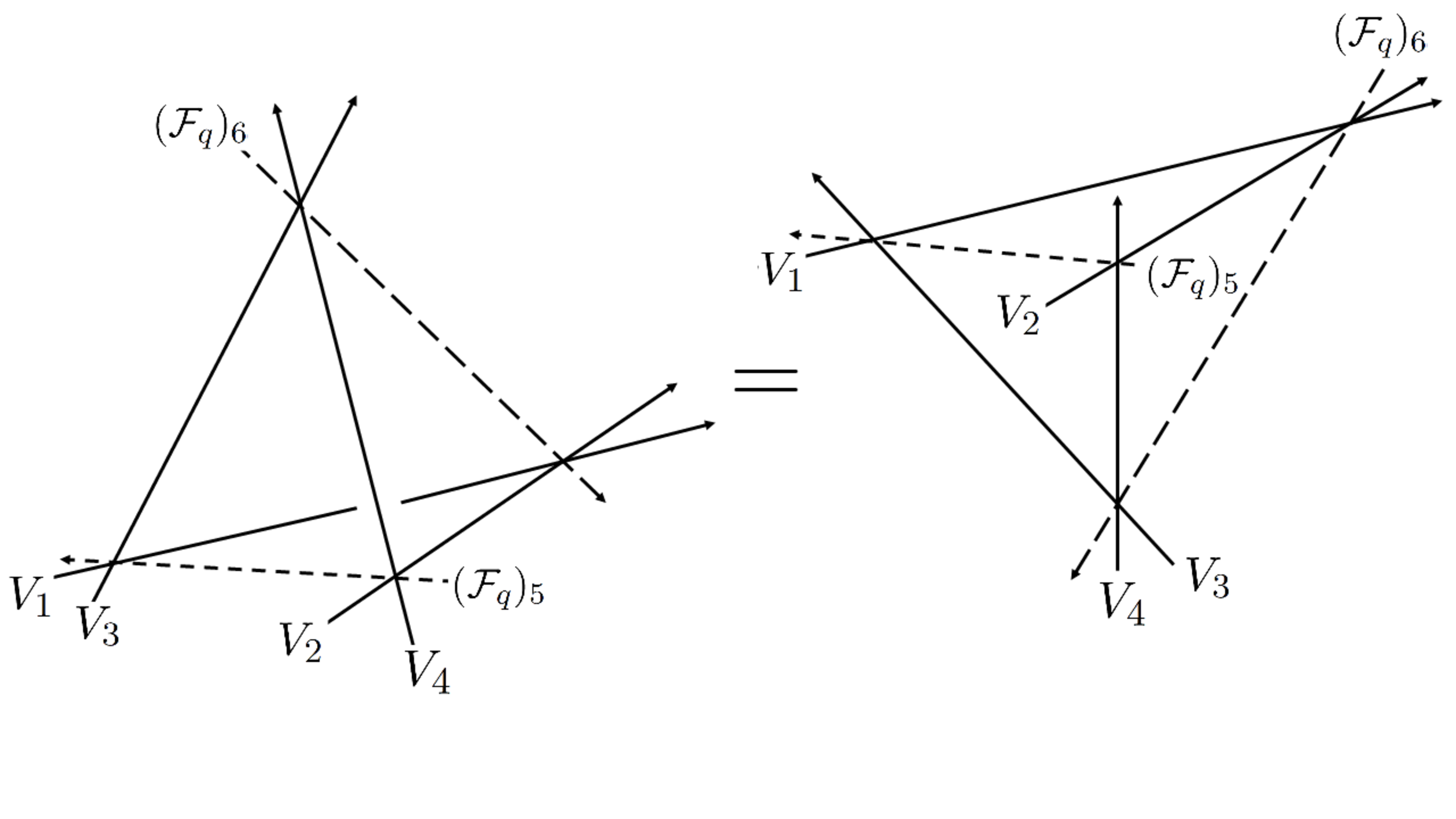}
\caption{
Tetrahedron equation.
The left and right figure represents
$\mathcal{M}_{126}(z_{12})\mathcal{M}_{346}(z_{34})\mathcal{L}_{135}(z_{13})\mathcal{L}_{245}(z_{24})$
and
$
\mathcal{L}_{245}(z_{24}) \mathcal{L}_{135}(z_{13}) \mathcal{M}_{346}(z_{34}) \mathcal{M}_{126}(z_{12}),
$ respectively.
}
\label{Tetrahedronequationfigure}
\end{figure}

We record the tetrahedron equation and the original
three-dimensional $R$-matrix for completeness.
See \cite{BaSe,Kuniba,KMO2} for more details.
\eqref{q=0L} is the $q=0$ degeneration of
the following version
\begin{align}
\mathcal{L}(z)( v_i \otimes v_j)
=\sum_{a=0}^1 \sum_{b=0}^1  v_a \otimes v_b [\mathcal{L}(z)]_{ij}^{ab},
\nn
\end{align}
where
\begin{align}
[\mathcal{L}(z)]_{00}^{00}&=1, \ \ \
[\mathcal{L}(z)]_{11}^{11}=1, \ \ \
[\mathcal{L}(z)]_{10}^{01}=z \mathbf{a}^+, \nn \\
[\mathcal{L}(z)]_{01}^{10}&=z^{-1} \mathbf{a}^-, \ \ \
[\mathcal{L}(z)]_{01}^{01}=\mathbf{k}, \ \ \ 
[\mathcal{L}(z)]_{10}^{10}=-q \mathbf{k}, \nn
\end{align}
and $[\mathcal{L}(z)]_{ij}^{ab}=0$ otherwise.
Here, $\mathbf{a}^+$, $\mathbf{a}^-$, $\mathbf{k}$,
together with $\mathbf{k}^{-1}$ form the $q$-Oscillator algebra
whose defining relations are
\begin{align}
    &\mathbf{k} \mathbf{a}^+ = q \mathbf{a}^+ \mathbf{k},
    \ \ \ \mathbf{k} \mathbf{a}^- = q^{-1} \mathbf{a}^- \mathbf{k},
    \ \ \ \mathbf{a}^- \mathbf{a}^+ =1-q^2 \mathbf{k}^2, \nn \\
    & \mathbf{a}^+ \mathbf{a}^- =1- \mathbf{k}^2, \ \ \
    \mathbf{k} \mathbf{k}^{-1} = \mathbf{k}^{-1} \mathbf{k}=1,
    \nn
\end{align}
The operators act on
the basis of
$q$-bosonic Fock space $\mathcal{F}_q=\displaystyle \oplus_{m=0}^\infty \mathbb{C}(q) |m \rangle$ as
\begin{align}
\mathbf{k}|m \rangle = q^m|m \rangle, \ \ \
\mathbf{a}^+|m \rangle =|m+1 \rangle, \ \ \
\mathbf{a}^-|m \rangle =(1-q^{2m})|m-1 \rangle. \nn
\end{align}
Define also
$\mathcal{M}(z):=\mathcal{L}(z)|_{q \to -q}$.
$\mathcal{L}(z)$ and $\mathcal{M}(z)$ can be viewed as operators acting on $V \otimes V \otimes \mathcal{F}_q$.
We can naturally extend this to operators acting on the tensor product
of many more two-dimensional and Fock spaces
by acting on three spaces nontrivially and the other spaces as identity.
It is customary to indicate the three spaces as subscripts.
$\mathcal{L}_{ijk}(z)$ and $\mathcal{M}_{ijk}(z)$
imply they act on $V_i$, $V_j$ and $(\mathcal{F}_q)_k$ nontrivially
and the other spaces as identity.
Using these notations, the operators
satisfy the following tetrahedron equation (Figure
\ref{Tetrahedronequationfigure})
\cite{BaSe,Kuniba,KMO2}
\begin{align}
    &\mathcal{M}_{126}(z_{12})\mathcal{M}_{346}(z_{34})\mathcal{L}_{135}(z_{13})\mathcal{L}_{245}(z_{24})
    \nn \\
    =&\mathcal{L}_{245}(z_{24}) \mathcal{L}_{135}(z_{13}) \mathcal{M}_{346}(z_{34}) \mathcal{M}_{126}(z_{12}),
    \nn
\end{align}
where $z_{ij}=z_i/z_j$.
Both hand sides act on $V_1 \otimes V_2 \otimes V_3 \otimes V_4 \otimes
(\mathcal{F}_q)_5 \otimes (\mathcal{F}_q)_6$.
The Zamolodchikov-Faddeev algebra commutation relations
\eqref{FZalg} are derived by using
the tetrahedron equation
\cite{Kuniba,KMO2}.
The $q=0$ degeneration of $\mathcal{L}(z)$
in this Appendix gives the one used in section 4.
Note
$\lim_{q \to 0} \mathbf{a}^+ = \mathbf{b}^+,
\lim_{q \to 0} \mathbf{a}^- = \mathbf{b}^-,
\lim_{q \to 0} \mathbf{k} = \mathbf{t}$.

\section*{Appendix B: Figures of configurations}

We give several figures
in this appendix.
Figure \ref{tetrahedroninversefigure}
is an example of Lemma \ref{simplestpartition} which
corresponds to the case 
$\langle \Omega| X_1^{(4)}(z_1)X_2^{(4)}(z_2)X_3^{(4)}(z_3)X_3^{(4)}(z_4) X_4^{(4)}(z_5)|\Omega \rangle$.
The figure is the unique configuration.
The layer corresponding to
$X_1^{(4)}(z_1)$,
$X_2^{(4)}(z_2)$,
$X_3^{(4)}(z_3)$,
$X_3^{(4)}(z_4)$ and
$ X_4^{(4)}(z_5)$
contributes the weight $z_1$, $z_2^2$,
$z_3^3$, $z_4^3$ and $z_5^4$ respectively,
and multiplying them gives $z_1 z_2^2 z_3^3 z_4^3 z_5^4$.

\begin{figure}[htbp]
\centering
\includegraphics[width=12truecm]{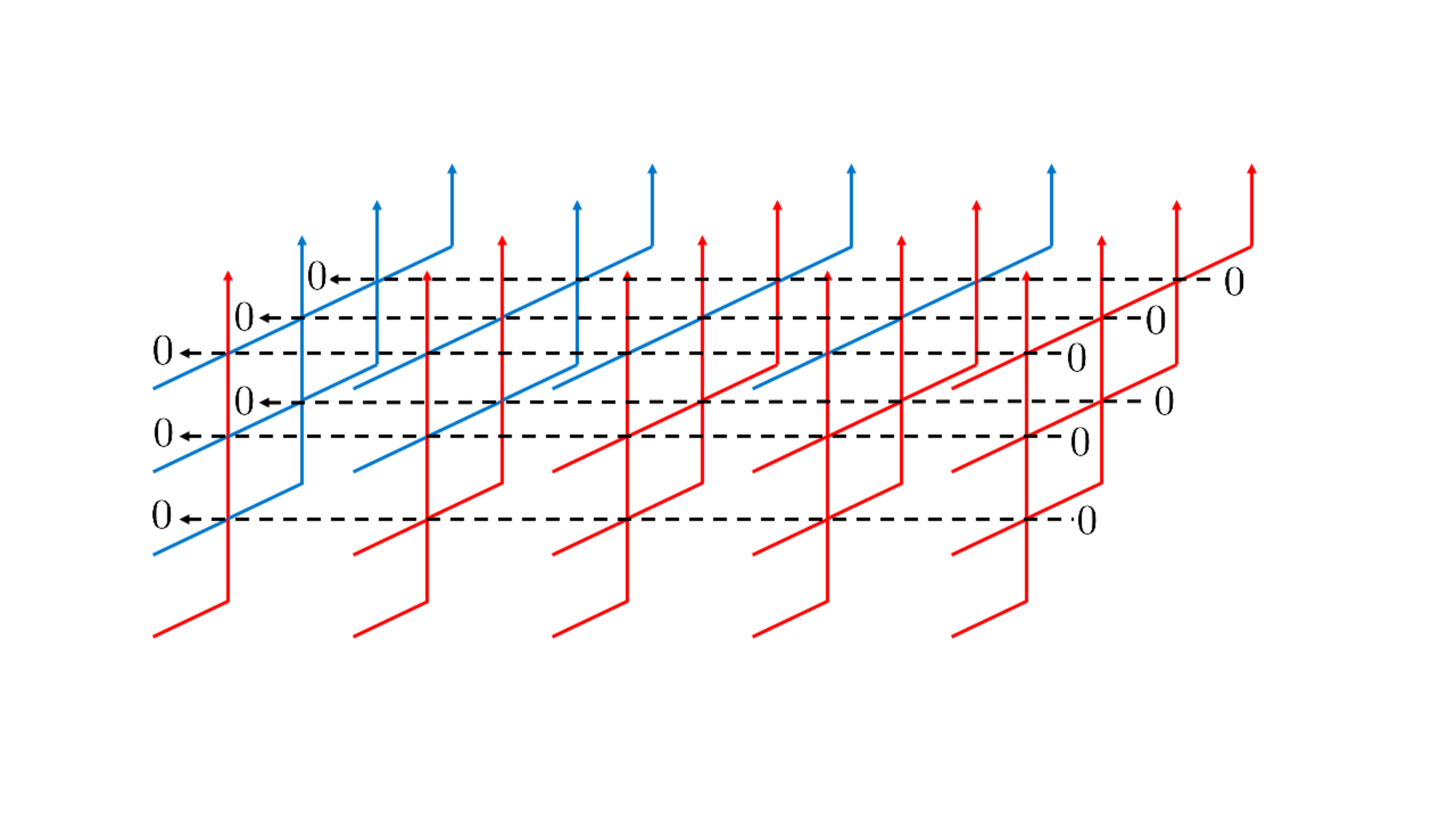}
\caption{
The unique nontrivial configuration for
$\langle \Omega| X_1^{(4)}(z_1)X_2^{(4)}(z_2)X_3^{(4)}(z_3)X_3^{(4)}(z_4) X_4^{(4)}(z_5)|\Omega \rangle$.
The layer corresponding to
$X_1^{(4)}(z_1)$,
$X_2^{(4)}(z_2)$,
$X_3^{(4)}(z_3)$,
$X_3^{(4)}(z_4)$ and
$ X_4^{(4)}(z_5)$
contributes the weight $z_1$, $z_2^2$,
$z_3^3$, $z_4^3$ and $z_5^4$ respectively,
and multiplying them gives $z_1 z_2^2 z_3^3 z_4^3 z_5^4$.
}
\label{tetrahedroninversefigure}
\end{figure}

Figure 
\ref{tetrahedronpartitionfunctionsconfigurationsfigure}
gives an example of checking Theorem \ref{mainthm}
which corresponds to the case
$\langle \Omega| X_3^{(4)}(z_1)X_3^{(4)}(z_2)X_1^{(4)}(z_3)|\Omega \rangle$. There are three nontrivial configurations,
and each of them contributes the weight
$z_1^3 z_2^2 z_3^2$,
$z_1^3 z_2^3 z_3$ and $z_1^2 z_2^3 z_3^2$ respectively.
Taking sum of the weights gives 
$z_1^3 z_2^2 z_3^2+z_1^3 z_2^3 z_3+z_1^2 z_2^3 z_3^2=
z_1 z_2 s_{(2,2,1)}(z_1,z_2,z_3)$.
\begin{figure}[htbp]
\centering
\begin{minipage}{.48\textwidth}
\includegraphics[width=\columnwidth]{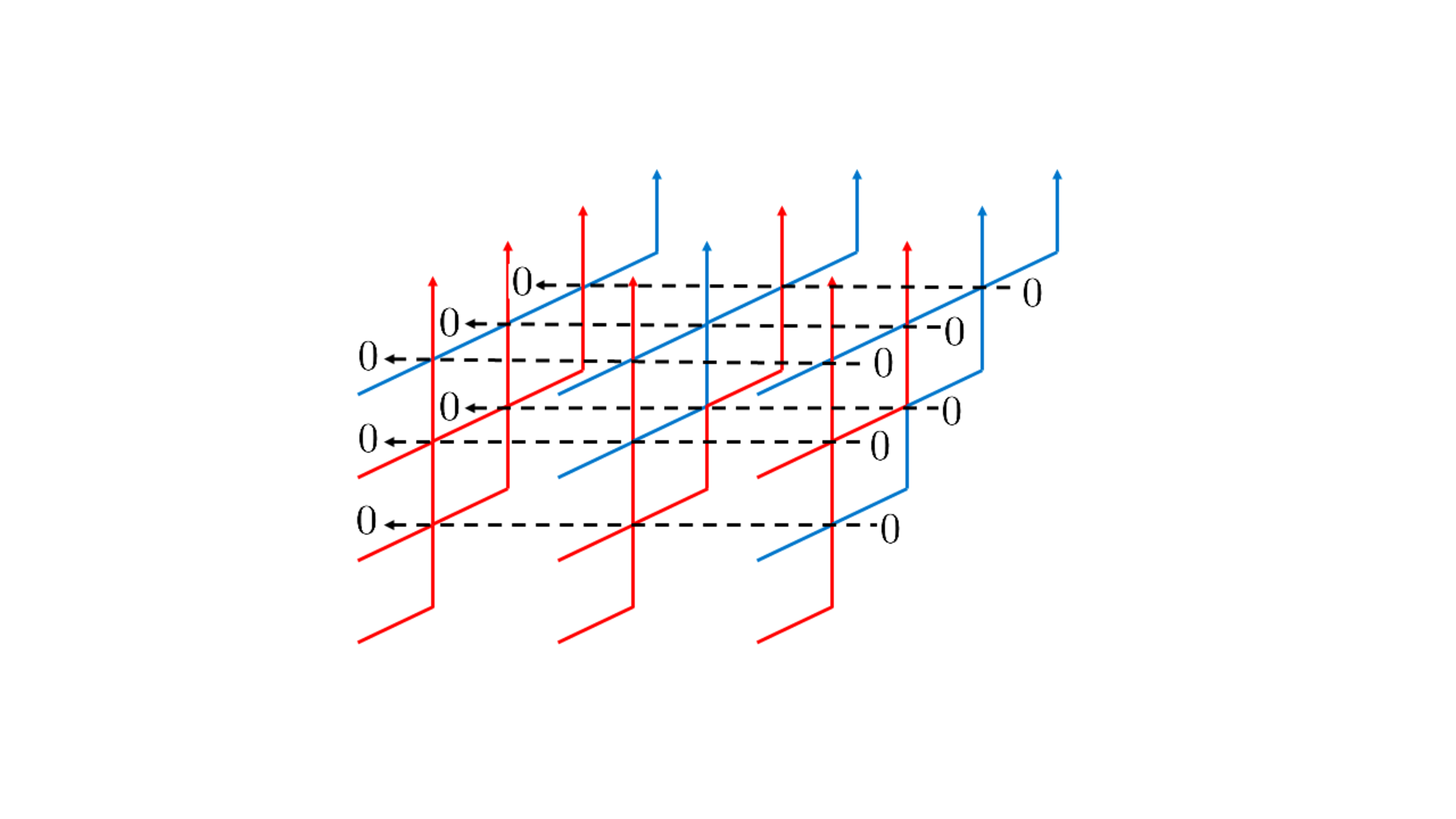}
\end{minipage}
\begin{minipage}{.48\textwidth}
\includegraphics[width=\columnwidth]{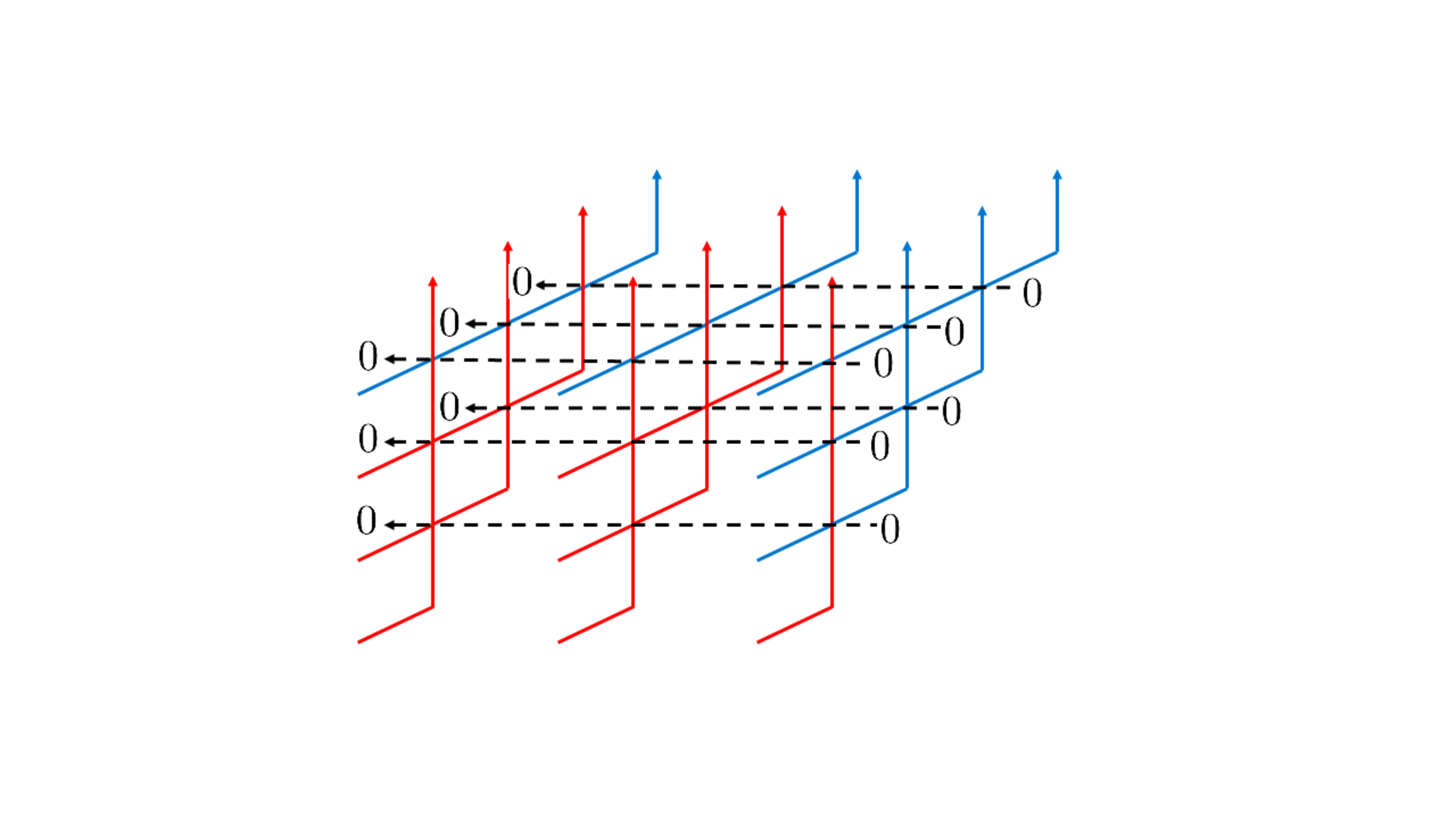}
\end{minipage}
\includegraphics[width=.48\columnwidth]{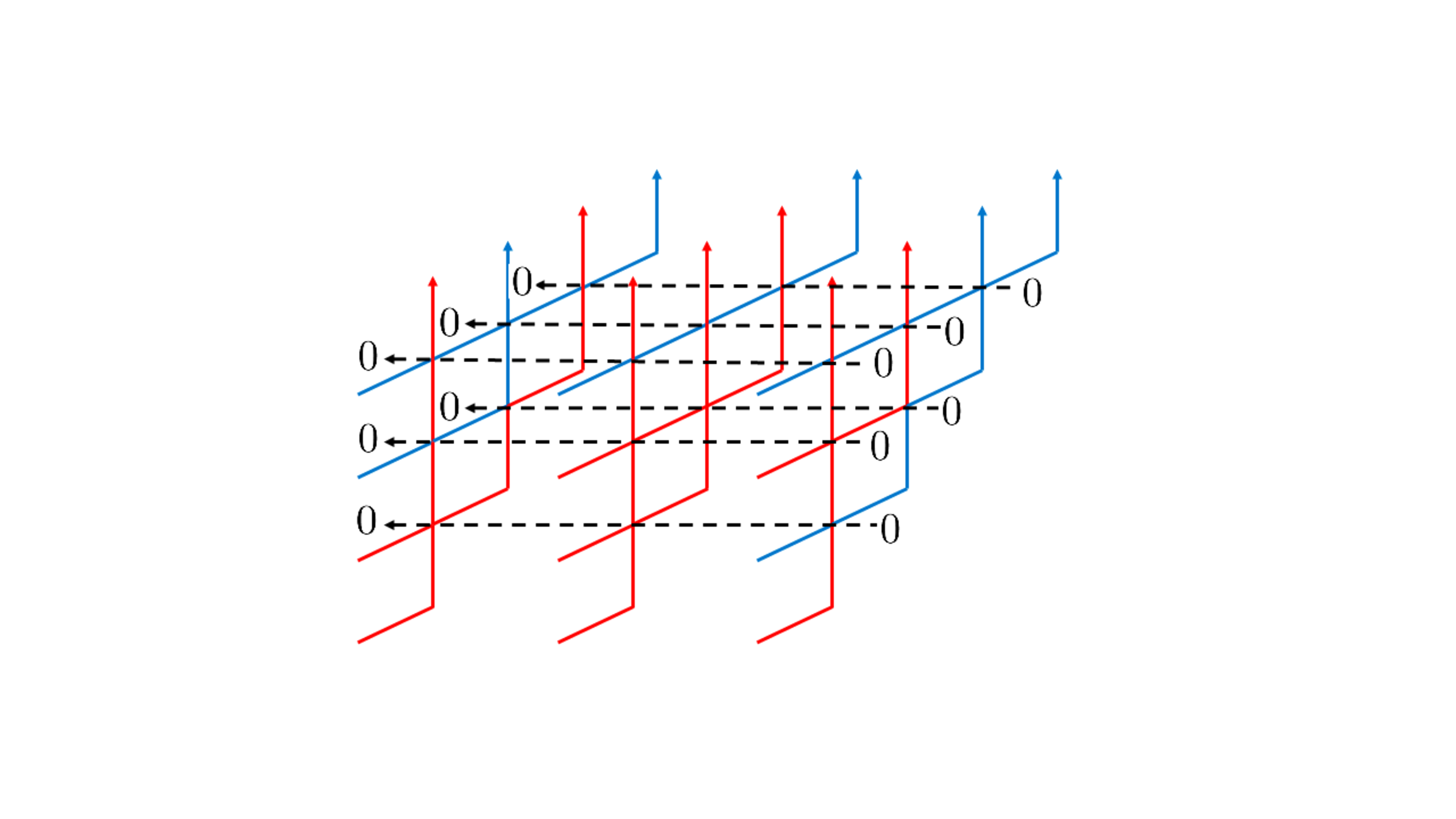}
\caption{
Three nontrivial
configurations for $\langle \Omega| X_3^{(4)}(z_1)X_3^{(4)}(z_2)X_1^{(4)}(z_3)|\Omega \rangle$.
Configuration in top left, top right and bottom panel
contributes the weight
$z_1^3 z_2^2 z_3^2$,
$z_1^3 z_2^3 z_3$ and $z_1^2 z_2^3 z_3^2$ respectively,
and taking sum gives $z_1 z_2 s_{(2,2,1)}(z_1,z_2,z_3)$.
}
\label{tetrahedronpartitionfunctionsconfigurationsfigure}
\end{figure}

\end{document}